\documentclass[manuscript,screen]{acmart}

\usepackage{algorithm}

\usepackage{graphicx} 
\usepackage{booktabs}
\usepackage{xcolor}
\usepackage{multirow}
\usepackage{makecell}
\usepackage{bm}
\usepackage{algpseudocode}  
\usepackage{amsmath}  
\usepackage[switch]{lineno}

\definecolor{cvprblue}{rgb}{0.21,0.49,0.74}
\definecolor{cYellow}{HTML}{FFFFCC}
\definecolor{cRed}{HTML}{FFCCCC} 
\definecolor{cGrey}{HTML}{F3F7F2} 
\definecolor{cGreen}{HTML}{339933}

\usepackage[dvipsnames]{xcolor}
\usepackage{colortbl}
\usepackage{subfig}
\usepackage{pifont}

\AtBeginDocument{%
  }

\setcopyright{acmlicensed}
\copyrightyear{2018}
\acmYear{2018}
\acmDOI{XXXXXXX.XXXXXXX}
\acmConference[Conference acronym 'XX]{Make sure to enter the correct
  conference title from your rights confirmation emai}{June 03--05,
  2018}{Woodstock, NY}
\acmISBN{978-1-4503-XXXX-X/18/06}

\begin{document}

\title{Who Stole Your Data? A Method for Detecting Unauthorized RAG Theft}

\author{Peiyang Liu}
\email{liupeiyang@pku.edu.cn}
\orcid{0000-0003-3658-9147}
\affiliation{%
  \institution{National Engineering Research Center for Software Engineering, Peking University}
  \city{Beijing}
  \country{China}
}

\author{Ziqiang Cui}
\email{ziqiang.cui@my.cityu.edu.hk}
\affiliation{%
  \institution{City University of Hong Kong}
  \city{Hong Kong}
  \country{China}
}

\author{Di Liang}
\email{liangd17@fudan.edu.cn}
\affiliation{%
  \institution{Fudan University}
  \city{Shanghai}
  \country{China}
}

\author{Wei Ye}
\authornote{Corresponding Author.}
\email{wye@pku.edu.cn}
\affiliation{%
  \institution{National Engineering Research Center for Software Engineering, Peking University}
  \city{Beijing}
  \country{China}
}

\renewcommand{\shortauthors}{Liu et al.}

\begin{abstract}
    Retrieval-augmented generation (RAG) is widely used to address hallucinations and outdated data issues in Large Language Models (LLMs). While RAG significantly enhances the credibility and real-time relevance of generated content by incorporating external data sources, it has also enabled unauthorized data appropriation at an unprecedented scale. Detecting such unauthorized data usage presents a critical challenge that demands innovative solutions.
    This paper addresses the challenge of detecting unauthorized RAG-based content appropriation through two key contributions. First, we introduce RPD, a novel dataset specifically designed for RAG plagiarism detection that overcomes the limitations of existing resources. RPD ensures recency, diversity, and realistic simulation of RAG-generated content across various domains and writing styles by modeling diverse professional domains and writing styles. Second, we propose a novel dual-layered watermarking system that sequentially embeds protection at both semantic and lexical levels, creating a robust detection mechanism that remains effective even after LLM reformulation. This system is complemented by an Interrogator-Detective framework that strategically probes suspected RAG systems and applies statistical hypothesis testing to accumulated evidence, enabling reliable detection even when individual watermark signals are weakened during the RAG process.
    Extensive experiments validate our approach's effectiveness across varying query volumes, defense prompts, and retrieval parameters. Our method demonstrates exceptional resilience against adversarial evasion techniques while maintaining high-quality text output. The dual-layered approach provides complementary protection mechanisms that address the fundamental challenges of RAG plagiarism detection, including LLM reformulation, fact redundancy, and adversarial evasion. This work provides a foundational framework for protecting intellectual property in the era of retrieval-augmented AI systems.
    Our source code and dataset are publicly available at \url{https://anonymous.4open.science/r/RPD-7E98}\text{.}
\end{abstract}



\received{20 February 2007}
\received[revised]{12 March 2009}
\received[accepted]{5 June 2009}

\maketitle

\section{Introduction}
Large Language Models have transformed how we generate and interact with textual content \cite{achiam2023gpt,liu2024deepseek,yang2025qwen3}. Despite their capabilities, these models face inherent limitations in accessing up-to-date information and specialized knowledge beyond their training data. Retrieval-Augmented Generation (RAG) \cite{lewis2020retrieval,gao2023retrieval} has emerged as an effective solution to this challenge by dynamically incorporating external knowledge sources into the generation process. While RAG has significantly enhanced LLM capabilities across numerous applications, it has simultaneously raised important ethical and legal concerns regarding intellectual property rights.

\subsection{The Challenge of Unauthorized Content Appropriation}
RAG systems can retrieve valuable information created by content producers, incorporate it into their outputs, and serve it to users without proper attribution or compensation. This raises significant questions about ownership rights and the traceability of information as it moves through these systems.

The unauthorized appropriation of content through RAG systems represents a significant challenge at the intersection of technology, ethics, and law. Content creators—ranging from news organizations and academic institutions to individual bloggers—invest substantial resources in producing high-quality information. When this content is seamlessly integrated into LLM outputs without attribution, it not only undermines the creators' economic interests but also erodes the foundation of intellectual property rights that incentivize knowledge creation \cite{ma2022specialization}.

\begin{figure*}[t]
    \centering
    \includegraphics[scale=0.5]{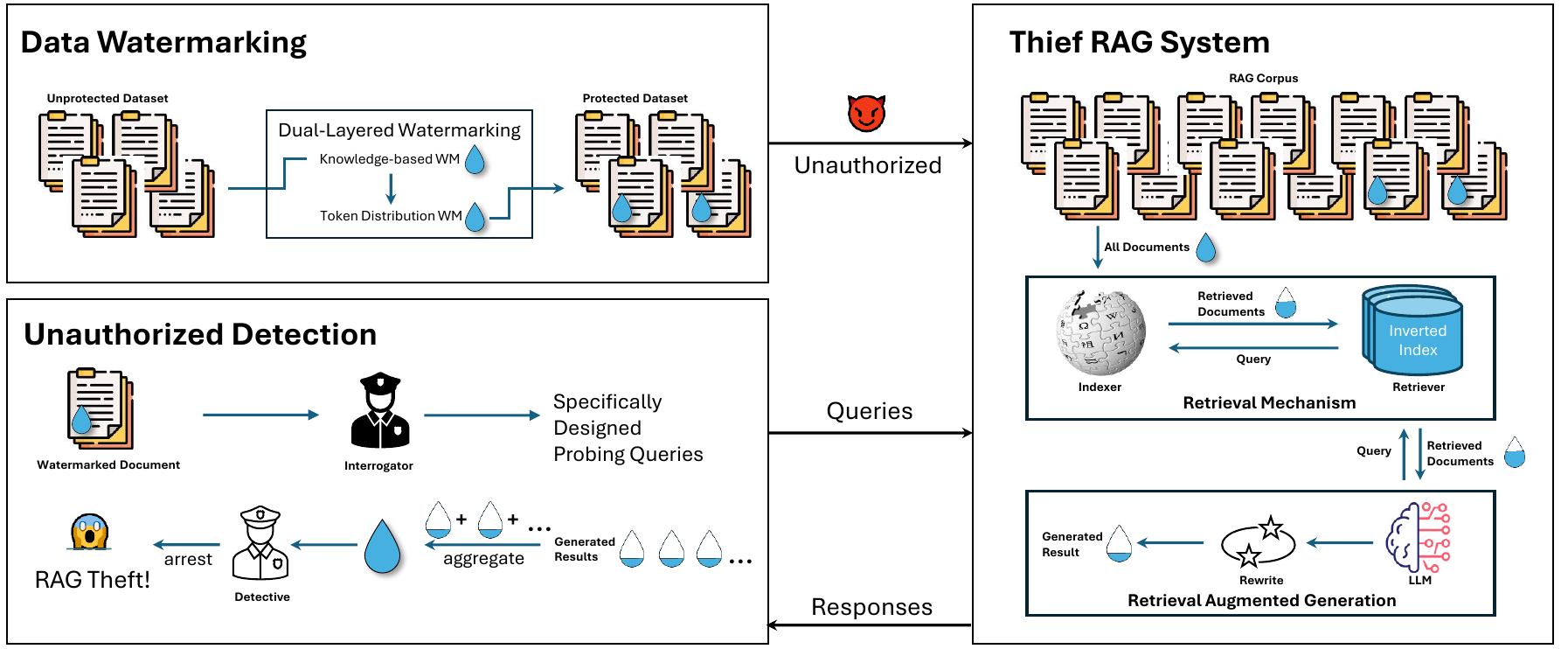}
    \caption{Overview of our proposed pipeline for preventing unauthorized RAG misuse. First, we apply our Dual-Layered Watermarking method to embed watermarks into the documents of the protected dataset. When a Thief RAG System illicitly accesses our dataset, the watermark propagates through the RAG pipeline into the LLM's generated outputs. Although the watermark inevitably weakens during propagation, in the detection phase, the Interrogator strategically crafts queries that increase the likelihood of the RAG System retrieving watermarked documents. By submitting multiple queries to the Thief RAG System, the aggregated watermark signals accumulate to a detectable level, enabling statistical hypothesis testing to determine whether unauthorized misuse has occurred.}
    \label{figure_framework}
\end{figure*}

\subsection{The Detection Challenge}
Detecting RAG-based content appropriation presents unique technical challenges that distinguish it from traditional plagiarism detection:

\begin{itemize}
    \item \textbf{LLM Reformulation:} Unlike direct copying, RAG systems typically use retrieved content as context for generation rather than verbatim reproduction \cite{yu2024evaluation}, resulting in paraphrased or synthesized outputs that maintain the core information while altering the linguistic surface.

    \item \textbf{Fact Redundancy:} Multiple sources may contain identical factual information, making it difficult to determine which specific source was used in the retrieval process \cite{jovanovic2025ward}.

    \item \textbf{Adversarial Evasion:} As detection methods emerge, malicious actors may develop countermeasures specifically designed to evade these techniques \cite{jovanovic2024watermarkstealing,lyu2025crud}.

\end{itemize}

Despite these challenges, the ability to detect unauthorized RAG usage is crucial for protecting intellectual property, ensuring proper attribution, and maintaining the integrity of the information ecosystem. Current approaches to this problem remain limited, with a notable absence of specialized datasets and detection methodologies tailored to the unique characteristics of RAG-based content appropriation.

\subsection{Towards Effective RAG Plagiarism Detection}
To address these challenges, we recognize the need for two fundamental components: (1) a specialized dataset that realistically captures the nuances of RAG-generated content across diverse domains and writing styles, and (2) a robust detection methodology capable of identifying unauthorized content use even after LLM reformulation.

Traditional watermarking methods for LLMs \cite{pmlr-v202-kirchenbauer23a,liu2024survey} are vulnerable to adversarial rewriting by illegal RAG systems, while existing fact-based techniques \cite{yang2023watermarking,chang-etal-2024-postmark} fail to address the data redundancy issues commonly found in RAG systems. Similarly, existing plagiarism detection datasets fail to account for the unique characteristics of RAG-based content appropriation, particularly the problems of fact redundancy caused by diverse writing styles among different creators in real-world scenarios \cite{jovanovic2025ward}.

Our approach introduces a novel dual-layered watermarking system that operates sequentially to create robust protection against unauthorized RAG usage. We develop a sophisticated framework that embeds distinctive knowledge-based watermarks at the semantic level while also creating statistical signatures through token distribution manipulation. This is complemented by an innovative Interrogator-Detective framework that strategically probes suspected RAG systems and applies statistical hypothesis testing to detect unauthorized usage. The overview of our method can be found in Figure \ref{figure_framework}. Our dataset creation methodology further enhances this approach by simulating the diversity and complexity of real-world RAG applications.

\subsection{Our Contributions}
This paper addresses the critical gap in RAG plagiarism detection through the following contributions:

\begin{enumerate}
    \item \textbf{Novel Dataset Creation:} We introduce a comprehensive dataset specifically designed for RAG plagiarism detection that overcomes the limitations of existing resources. Our approach leverages the \textit{repliqa} \cite{monteiro2024repliqa} foundation to ensure recency and employs a sophisticated author simulation system that models diverse professional domains, writing styles, and preference dimensions. This methodology produces content that realistically simulates the variability and redundancy found in real-world RAG scenarios.

    \item \textbf{Dual-Layered Watermarking System:} We propose an innovative sequential watermarking approach that embeds protection at both semantic and lexical levels. Our system strategically integrates knowledge-based watermarks with statistical token patterns to create a detection mechanism that:
    1. Remains effective even after LLM reformulation through the RAG pipeline.
    2. Addresses the fact redundancy problem through distinctive knowledge embedding
    3. Provides complementary layers of protection against various evasion techniques.

    \item \textbf{Interrogator-Detective Framework:} We develop a novel detection framework that strategically generates queries to reveal unauthorized data usage and applies rigorous statistical analysis to accumulated evidence, enabling reliable detection even when individual watermark signals are weakened during the RAG process.

    \item \textbf{Comprehensive Evaluation Framework:} We implement a rigorous testing methodology that evaluates our approach across multiple dimensions, including detection accuracy as a function of query volume, resilience against defensive prompting, impact on text quality, and sensitivity to parameter variations.

\end{enumerate}

\section{Related Work}
Our research builds upon and extends several key areas of prior work, including retrieval-augmented generation systems, watermarking techniques for language models, and plagiarism detection methodologies.
\subsection{Retrieval Augmented Generation}
Retrieval Augmented Generation (RAG) has emerged as a powerful paradigm for enhancing LLM capabilities by incorporating external knowledge. The seminal work by \cite{lewis2020retrieval} introduced the RAG framework, demonstrating significant improvements in knowledge-intensive tasks. Recent advancements have focused on optimizing retrieval mechanisms \cite{izacard2023atlas,shi2023replug} and improving the integration of retrieved information into generation processes \cite{asai2023self,zhou2023context}.

\cite{gao2023retrieval} provided a comprehensive survey of RAG techniques, highlighting their applications across various domains and identifying open challenges.
More recently, \cite{yu2024evaluation} proposed evaluation frameworks specifically designed for RAG systems, emphasizing the need for metrics that assess both retrieval quality and generation fidelity. \cite{jiang2023active} introduced active retrieval augmented generation, which dynamically determines when to retrieve information based on model uncertainty.

Despite these advancements, the ethical and legal implications of RAG systems have received comparatively less attention. \cite{ai2024artificial} discussed the potential for copyright infringement in retrieval-based systems, while \cite{henderson2023foundation} highlighted the attribution challenges when LLMs incorporate external knowledge sources.

\subsection{Watermarking Techniques for Language Models}
Watermarking techniques for language models have evolved significantly in recent years. Traditional approaches focused on embedding imperceptible signals within generated text to enable ownership verification and detection of AI-generated content.

\cite{pmlr-v202-kirchenbauer23a} introduced a ``red-green'' watermarking method that biases the token selection process during generation, creating statistical patterns detectable through hypothesis testing. This approach was further refined by \cite{kuditipudi2023robust}, who developed techniques to make watermarks more resilient against adversarial modifications. \cite{liu2024survey} provided a comprehensive overview of watermarking approaches for generative AI, categorizing methods based on their embedding strategies and detection capabilities.

More recently, fact-based watermarking has emerged as a promising direction. \cite{yang2023watermarking} proposed embedding factual modifications within generated text that serve as identifiable markers while maintaining semantic coherence. \cite{chang-etal-2024-postmark} introduced PostMark, a post-processing watermarking technique that preserves the original generation quality while enabling reliable detection. \cite{christ2024undetectable} explored the theoretical limits of watermarking, identifying fundamental trade-offs between detectability and text quality.

However, as noted by \cite{jovanovic2024watermarkstealing}, existing watermarking techniques remain vulnerable to sophisticated evasion strategies, particularly in the context of RAG systems where content undergoes significant transformation during the retrieval and generation process.

\subsection{Plagiarism Detection in AI-Generated Content}
Traditional plagiarism detection methods have primarily focused on identifying verbatim copying or paraphrasing in human-written text. The emergence of AI-generated content has necessitated new approaches that can identify more sophisticated forms of content appropriation.

\cite{sadasivan2023can} investigated the detectability of LLM-generated text, finding that current detection methods struggle with high-quality generations from advanced models. \cite{mitchell2023detectgpt} proposed DetectGPT, a reference-free approach for identifying text generated by large language models based on curvature in the model's log probability space. \cite{krishna2023paraphrasing} examined how paraphrasing affects the detectability of AI-generated text, revealing significant challenges for current detection systems. \cite{anderson2024my} proposes a membership inference attack method against RAG systems, demonstrating its effectiveness across datasets and models while suggesting preliminary defense strategies.
\cite{pang2023black} introduces a black-box membership inference attack framework for fine-tuned diffusion models, achieving high effectiveness across multiple attack scenarios with an AUC of up to 0.95.

In the specific context of RAG systems, \cite{jovanovic2025ward} highlighted the problem of fact redundancy, where multiple sources contain identical factual information, complicating the attribution process. \cite{lyu2025crud} addressed the scalability challenges of detection methods, proposing efficient algorithms for processing large volumes of potentially appropriated content.

\begin{figure*}[t]
	    \centering
	    \includegraphics[scale=0.185]{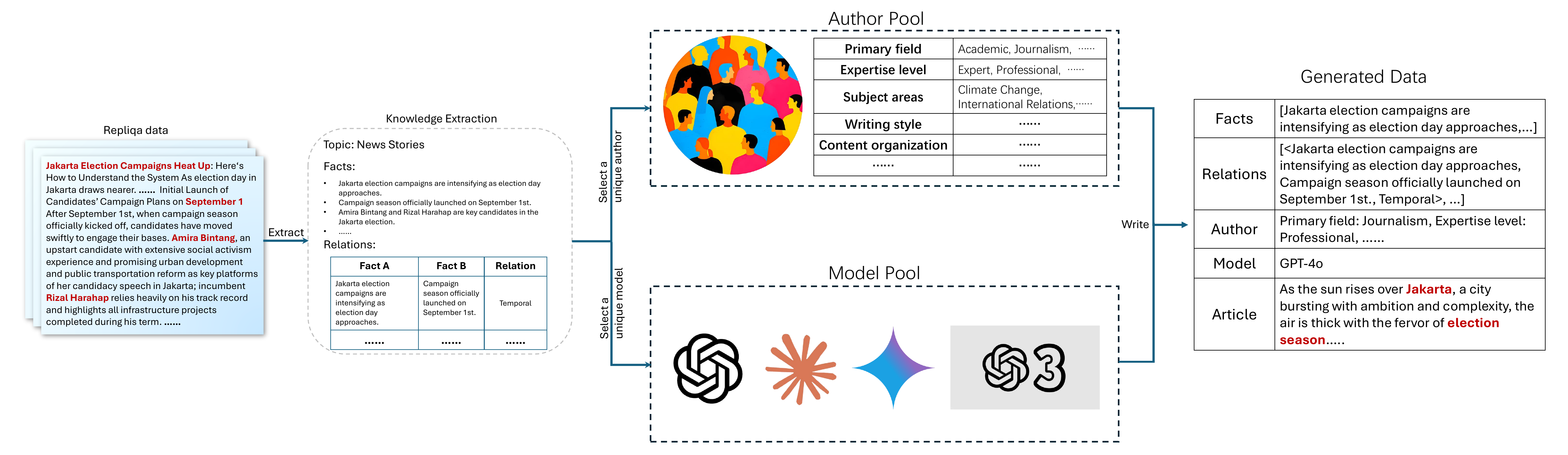}
	    \caption{The construction flowchart of the RPD dataset. The source data is derived from the repliqa dataset, ensuring that the original data does not appear in the training samples of LLMs to prevent data leakage. For each piece of data in the repliqa dataset, we extract Facts and Relations. Based on these facts and relations, we select a unique author role from an ``Author Pool'' composed of writers with diverse styles. Different LLMs then assume these distinct author roles to select Facts and Relations for writing articles. This approach effectively simulates the data redundancy issues caused by authors with varying styles in real-world scenarios. Articles written based on true facts and relations will not conflict with the knowledge already learned by LLMs, and the newly generated data can avoid data leakage.}
	    \label{figure_datagenerate}
\end{figure*}

\section{RAG Plagiarism Detection Dataset}

Current plagiarism detection datasets suffer from several critical limitations when applied to RAG scenarios. First, they typically focus on verbatim copying or simple paraphrasing, whereas RAG systems fundamentally transform content through LLM reformulation while preserving core information. Second, they rarely account for the fact redundancy problem—where multiple legitimate sources may contain identical factual information, making source attribution challenging. Finally, most existing datasets lack the diversity in writing styles, domain expertise, and content organization that characterizes real-world information ecosystems.
To address these issues, we provide a new \textbf{R}AG \textbf{P}lagiarism \textbf{D}etection dataset, named \textbf{RPD}. 

\subsection{Dataset Design Philosophy}

Our dataset design addresses these limitations through a novel approach that simulates the complex interplay between content creators, information structures, and linguistic expression. We built our dataset on three foundational principles:

\begin{enumerate}
    \item \textbf{Factual Consistency with Stylistic Diversity:} Content should maintain consistent core facts while varying significantly in expression across different authors and writing styles.
    
    \item \textbf{Controlled Redundancy:} The dataset should contain deliberate overlaps in factual information across documents, mirroring real-world information ecosystems where multiple sources cover similar topics.
    
    \item \textbf{Recency and Relevance:} To avoid the problem of LLMs recognizing content from their training data, the dataset should contain recent information not included in model training.
\end{enumerate}

\subsection{Dataset Generation Methodology}

To implement these principles, we developed a multi-stage generation process that transforms source documents into a diverse collection of stylistically varied but factually consistent articles.

\subsubsection{Source Selection and Fact Extraction}

We began with the \textit{repliqa} dataset \cite{monteiro2024repliqa}, which provides recent documents across diverse domains. This choice was deliberate—\textit{repliqa} contains information published after the training cutoff dates of current LLMs, ensuring our dataset represents genuinely novel content that models cannot have memorized during pre-training.

As shown in Figure \ref{figure_prompts_fact_author} and \ref{figure_knowledge_example}, from each source document, we extracted two categories of facts:

\begin{itemize}
    \item \textbf{Core Facts:} 5-7 essential pieces of information that capture the document's central claims or findings.
    \item \textbf{Extended Facts:} 8-12 supporting details that provide context or elaboration.
\end{itemize}

Additionally, we identified semantic relationships between facts (causal, temporal, supportive, contradictory, or elaborative), creating a knowledge graph representation of each document. This structured approach allows us to maintain factual consistency while enabling stylistic variation in the regenerated content.

\subsubsection{Author Profile Simulation}

To capture the diversity of real-world content creation, we developed a sophisticated author simulation system with an author pool containing 100 different author styles (as described in Figure \ref{figure_prompts_fact_author} and \ref{figure_author_example}). Each simulated author is defined along multiple dimensions:

\begin{itemize}
    \item \textbf{Domain Expertise:} Primary field (academic, journalism, technical writing, etc.), expertise level, and specific subject areas.
    \item \textbf{Writing Style:} Formality, technical density, narrative approach, sentence structure, and vocabulary preferences.
    \item \textbf{Content Organization:} Structure preferences, detail level, and citation style.
    \item \textbf{Perspective Biases:} Political leaning, technology attitude, and epistemological approach.
\end{itemize}
This multidimensional characterization allows us to generate content that varies not just in surface-level expression but in deeper structural and stylistic elements—mirroring the diversity of real-world content while maintaining factual consistency.

\subsubsection{Multi-Model Content Generation}

To further enhance diversity and realism, we employed multiple state-of-the-art language models (GPT-4o, Claude 3.5 Sonnet, Gemini 2.0, and OpenAI o3) to generate content based on the extracted facts and author profiles. Each model brings its own linguistic patterns and generation tendencies, creating a dataset with natural variation that better represents the heterogeneity of real-world content.

For each source document, we generated multiple derivative articles that:

\begin{itemize}
    \item Incorporate all core facts and a subset of extended facts (specifically, 2 randomly selected extended facts)
    \item Maintain factual accuracy while varying expression
    \item Adhere to the assigned author profile's stylistic characteristics (each LLM randomly selects 1 style from the author pool)
    \item Include original examples, hypothetical scenarios, or contextual information consistent with the facts
\end{itemize}

This approach creates a controlled environment where we know precisely which facts appear in which documents, enabling accurate evaluation of the information sourcing behaviors of the RAG systems. More details can be found in Figure \ref{figure_prompts_article} and \ref{figure_article_example}.

\subsection{Dataset Characteristics}

\begin{table}[t]
\centering
\begin{tabular}{lr}
\hline
\textbf{Characteristic} & \textbf{Value} \\
\hline
Source documents (from \textit{repliqa}) & 3,000 \\
Core facts per document & 5-7 \\
Extended facts per document & 8-12 \\
Author styles in pool & 100 \\
LLMs used as authors & 4 \\
Extended facts selected per generation & 2 \\
Articles per source document & 4 \\
Total dataset entries & 12,000 \\
\hline
\end{tabular}
\caption{RPD Dataset Statistics}
\label{tab:dataset_stats}        
\end{table}

The resulting RPD dataset comprises 3,000 original documents from Repliqa, each with 4 derivative articles written in different styles (one per LLM author), resulting in a total of 12,000 entries as shown in Table \ref{tab:dataset_stats}. This structure creates a realistic simulation of the information ecosystem that RAG systems typically navigate, where multiple sources may contain overlapping information expressed in different ways.

\section{Methodology}

Our approach to detecting unauthorized RAG usage draws inspiration from this rich history of watermarking, adapted to the unique challenges of neural text generation. We propose a comprehensive framework consisting of two main components: (1) a dual-layered watermarking system for content protection and (2) an interrogator-detecting framework for unauthorized usage detection.

\subsection{Dual-Layered Watermarking System}

The core innovation of our methodology lies in its sequential dual-layered approach, applying two complementary watermarking techniques in succession to create a robust detection system that is resilient to various evasion attempts.

\subsubsection{Layer 1: Knowledge-Based Watermarking}

The first layer embeds distinctive knowledge information within the content, creating what we call ``knowledge watermarks''. Unlike traditional watermarks that modify the medium itself, knowledge watermarks operate at the semantic level, inserting carefully crafted information that serves as a signature.

The process begins by identifying candidate knowledge that could serve as watermarks. For a given document $D$ with a set of original facts $F_D = \{f_1, f_2, ..., f_n\}$, we construct a semantic embedding space $\mathcal{E}$ where each knowledge $f_i$ is mapped to a embedding $\mathbf{e}_i = \text{Embed}(f_i) \in \mathbb{R}^d$.

To identify potential watermark knowledge, we compute the semantic similarity between knowledge using cosine distance:
\begin{equation}
\text{sim}(f_i, f_j) = \frac{\mathbf{e}_i \cdot \mathbf{e}_j}{||\mathbf{e}_i|| \cdot ||\mathbf{e}_j||}.
\end{equation}
For each knowledge $f_i$, we identify a set of $k$ semantically similar knowledge $\mathcal{N}_k(f_i)$ that satisfy:
\begin{equation}
\mathcal{N}_k(f_i) = \{f_j \in \mathcal{F} \setminus F_D : \text{sim}(f_i, f_j) > \tau_{\text{sim}} \wedge \text{sim}(f_i, f_j) < \tau_{\text{ident}}\},
\end{equation}
where $\mathcal{F}$ represents our universal knowledge corpus, $\tau_{\text{sim}}$ is the minimum similarity threshold, and $\tau_{\text{ident}}$ is the maximum similarity threshold to avoid knowledge that are essentially identical. An example of sampled knowledge is shwon in Figure \ref{figure_sampled_knowledge_example}.

From these candidates, we select a subset of $m$ knowledge $W_D = \{w_1, w_2, ..., w_m\} \subset \bigcup_{i=1}^{n} \mathcal{N}_k(f_i)$ that maximize both semantic coherence with the document and distinctiveness as watermarks:
\begin{equation}
W_D = \mathrm{argmax}_{W \subset \bigcup_{i} \mathcal{N}_k(f_i), |W|=m} (\text{Coherence}(W, D)
+ \text{Distinctiveness}(W) ).
\end{equation}
As shown in Figure \ref{figure_prompts_facts_watermark}, the Coherence and Distinctiveness scores are evaluated by the Interrogator, which assesses how naturally the candidate watermark knowledge fits within the document context (Coherence) and how uniquely identifiable it would be as a watermark (Distinctiveness). This Agent-based evaluation allows us to select watermarks that maintain document integrity while providing strong detection signals.

These selected watermark knowledge are then strategically integrated into the document using a controlled language generation process that maintains natural flow and coherence, resulting in a knowledge-watermarked document $D_{\text{knowledge}}$.

\subsubsection{Layer 2: Red-Green Token Distribution Manipulation}
After embedding knowledge-based watermarks, we apply our second layer of protection, which works at the lexical level by subtly manipulating the statistical distribution of tokens in the generated text. This approach, inspired by the red-green watermarking technique introduced by \cite{pmlr-v202-kirchenbauer23a}, creates a statistical signature that persists even when the text is paraphrased.

For a vocabulary $\mathcal{V}$ and a context window $c_t = [t_{t-k}, ..., t_{t-1}]$ preceding position $t$, we partition the vocabulary into ``green'' tokens $\mathcal{G}(c_t)$ and ``red'' tokens $\mathcal{R}(c_t) = \mathcal{V} \setminus \mathcal{G}(c_t)$ using a seeding function $h$:

\begin{equation}
\mathcal{G}(c_t) = \{v \in \mathcal{V} : h(v, c_t) \mod \gamma < \gamma/2\},
\end{equation}
where $\gamma$ is a parameter controlling the proportion of green tokens.

We apply this second layer of watermarking to the knowledge-watermarked document $D_{\text{knowledge}}$ by regenerating it with a bias toward green tokens. During text generation, we modify the logits of the language model:

\begin{equation}
\tilde{p}(v|c_t) =
\begin{cases}
\frac{p(v|c_t) \cdot e^{\delta}}{Z(c_t)} & \text{if } v \in \mathcal{G}(c_t), \\
\frac{p(v|c_t)}{Z(c_t)} & \text{if } v \in \mathcal{R}(c_t),
\end{cases}
\end{equation}
where $p(v|c_t)$ is the original probability, $\delta$ is the bias strength, and $Z(c_t)$ is a normalization factor:

\begin{equation}
Z(c_t) = \sum_{v \in \mathcal{G}(c_t)} p(v|c_t) \cdot e^{\delta} + \sum_{v \in \mathcal{R}(c_t)} p(v|c_t).
\end{equation}
This creates a subtle statistical pattern where green tokens appear more frequently than expected by chance, yet the text remains fluent and natural to human readers. The result is a dual-watermarked document $D_{\text{dual}}$ that contains both knowledge-based and red-green watermarks.

\subsection{Interrogator-Detective Framework for Detection}

A critical challenge in detecting unauthorized RAG usage is that watermark signals become progressively weakened after retrieval and LLM reformulation. To address this challenge, we introduce an Interrogator-Detective framework that works in tandem to reveal unauthorized data usage.

\subsubsection{The Interrogator: Strategic Query Generation}

The Interrogator module generates strategic queries designed to reveal whether a RAG system has access to watermarked content:

\begin{equation}
Q_D = \text{Interrogator}(D_{\text{dual}}, W_D),
\end{equation}
where $Q_D$ is the generated query, $D_{\text{dual}}$ is the dual-watermarked document, and $W_D$ represents the knowledge-based watermarks.

The Interrogator specifically crafts questions that can only be answered correctly by accessing the knowledge-based watermarked content. These questions are designed with a critical property: they appear answerable to a RAG system that has indexed the watermarked content, but would force a system without access to respond with ``I don't know'' or provide an incorrect answer.

Mathematically, we define the ideal interrogation query $Q^*$ as:

\begin{equation}
    Q^* = \mathrm{argmax}_{Q \in \mathcal{Q}} (P(A_{\text{correct}} | \text{RAG with } D_{\text{dual}}, Q)
    - P(A_{\text{correct}} | \text{RAG without } D_{\text{dual}}, Q)),
\end{equation}
where $\mathcal{Q}$ is the space of possible queries, and $P(A_{\text{correct}} | \text{condition}, Q)$ is the probability of a correct answer given the condition and query.

The Interrogator's effectiveness stems from its ability to target the unique combination of knowledge watermarks embedded in the document, creating queries that are highly specific to watermarked content while appearing natural and relevant to the document's topic.

\subsubsection{The Detective: Statistical Evidence Analysis}

The Detective module analyzes the responses from the suspected RAG system to determine whether it has accessed watermarked content. Rather than making determinations based on individual samples, the Detective employs a statistical hypothesis testing framework:

\begin{itemize}
    \item $H_0$: The RAG system did not use the watermarked dataset.
    \item $H_1$: The RAG system plagiarized the watermarked dataset.
\end{itemize}

The Detective analyzes both layers of watermarking evidence:

\paragraph{Knowledge-Based Watermark Detection}
For a set of $n$ queries $\{Q_1, Q_2, ..., Q_n\}$ generated by the Interrogator, the Detective computes:

\begin{equation}
S_{\text{fact}} = \frac{1}{n} \sum_{i=1}^{n} \mathbb{I}(R_i \text{ contains watermarked knowledge}),
\end{equation}
where $R_i$ is the RAG system's response to query $Q_i$, and $\mathbb{I}(\cdot)$ is the indicator function.

Under the null hypothesis $H_0$ (no plagiarism), $S_{\text{fact}}$ follows a binomial distribution with parameters $n$ and $p_0$, where $p_0$ is the probability of correctly answering by chance. The Detective rejects $H_0$ if:

\begin{equation}
S_{\text{fact}} > \tau_{\text{fact}} = p_0 + z_{\alpha} \sqrt{\frac{p_0(1-p_0)}{n}},
\end{equation}
where $z_{\alpha}$ is the critical value corresponding to significance level $\alpha$.

\paragraph{Red-Green Token Distribution Detection}
For the same set of responses, the Detective analyzes the statistical distribution of tokens:

\begin{equation}
S_{\text{token}} = \frac{\sum_{i=1}^{n} \sum_{t \in R_i} \mathbb{I}(t \in \mathcal{G}(c_t))}{\sum_{i=1}^{n} |R_i|},
\end{equation}
where $|R_i|$ is the length of response $R_i$ in tokens.

Under $H_0$, $S_{\text{token}}$ should be approximately 0.5 (equal distribution of red and green tokens). The Detective rejects $H_0$ if:

\begin{equation}
Z_{\text{token}} = \frac{S_{\text{token}} - 0.5}{\sqrt{\frac{0.25}{\sum_{i=1}^{n} |R_i|}}} > z_{\alpha}.
\end{equation}

\paragraph{Combined Detection Decision}
The Detective combines evidence from both watermarking layers to make a final determination:

\begin{equation}
\text{Decision} =
\begin{cases}
\text{RAG Theft!} & \text{if } (S_{\text{fact}} > \tau_{\text{fact}}) \vee (Z_{\text{token}} > z_{\alpha}), \\
\text{Innocent} & \text{otherwise}.
\end{cases}
\end{equation}

This dual-criteria approach significantly reduces both false positives and false negatives compared to single-layer detection methods, as each layer provides complementary evidence of unauthorized usage.
An example of a dual-layer watermarked article is shown in Figure \ref{figure_watermarked_article_example}, and
the query specifically generated for dual-layer watermarked the article is shown in Figure \ref{figure_query_example}.

\section{Mathematical Analysis of the Dual-Layered Watermarking System}
In this section, we provide a rigorous mathematical analysis of our dual-layered watermarking approach for detecting unauthorized RAG usage. Our method combines knowledge-based watermarking with token distribution manipulation to create a robust detection framework.

\subsection{Analysis of Token Distribution Watermarking}

Following the red-green watermarking approach, we partition the vocabulary $\mathcal{V}$ into ``green'' tokens $\mathcal{G}(c_t)$ and ``red'' tokens $\mathcal{R}(c_t) = \mathcal{V} \setminus \mathcal{G}(c_t)$ using a seeding function $h$:
\begin{equation}
\mathcal{G}(c_t) = \{v \in \mathcal{V} : h(v, c_t) \mod \gamma < \gamma/2\},
\end{equation}
where $\gamma$ controls the proportion of green tokens and $c_t = [t_{t-k}, ..., t_{t-1}]$ is the context window preceding position $t$.
During text generation, we modify the logits of the language model to bias toward green tokens:
\begin{equation}
\tilde{p}(v|c_t) =
\begin{cases}
\frac{p(v|c_t) \cdot e^{\delta}}{Z(c_t)} & \text{if } v \in \mathcal{G}(c_t), \\
\frac{p(v|c_t)}{Z(c_t)} & \text{if } v \in \mathcal{R}(c_t),
\end{cases}
\end{equation}
where $p(v|c_t)$ is the original probability, $\delta$ is the bias strength, and $Z(c_t)$ is a normalization factor:
\begin{equation}
Z(c_t) = \sum_{v \in \mathcal{G}(c_t)} p(v|c_t) \cdot e^{\delta} + \sum_{v \in \mathcal{R}(c_t)} p(v|c_t).
\end{equation}

\subsection{Detection Sensitivity Analysis}
Let $\alpha = e^{\delta}$. For a document with token sequence $s = [s_1, s_2, ..., s_T]$, we define $|s|_G$ as the number of green tokens in the sequence. We can analyze the expected number of green tokens and its variance.
\begin{theorem}
For a watermarked text sequence of $T$ tokens generated with bias strength $\delta$ and green list proportion $\gamma$, if the average spike entropy of the sequence is at least $S^*$, then:

\begin{equation}
\mathbb{E}[|s|_G] \geq \frac{\gamma\alpha T}{1+(\alpha-1)\gamma} S^*
\end{equation}

Furthermore, the variance of the green token count is bounded by:
\begin{equation}
\text{Var}[|s|_G] \leq T \frac{\gamma\alpha S^*}{1+(\alpha-1)\gamma} \left(1-\frac{\gamma\alpha S^*}{1+(\alpha-1)\gamma}\right)
\end{equation}
where the spike entropy $S(p,z)$ of a probability distribution $p$ with modulus $z$ is defined as:

\begin{equation}
S(p,z) = \sum_k \frac{p_k}{1+zp_k}
\end{equation}
\end{theorem}

\begin{proof}
For each token position $t$, the probability of selecting a green token is:

\begin{equation}
P(s_t \in \mathcal{G}) = \frac{\sum_{v \in \mathcal{G}(c_t)} p(v|c_t) \cdot e^{\delta}}{Z(c_t)}
\end{equation}

Following the analysis in the reference paper, we can show that:

\begin{equation}
P(s_t \in \mathcal{G}) \geq \frac{\gamma\alpha}{1+(\alpha-1)\gamma} S\left(p,\frac{(1-\gamma)(\alpha - 1)}{1+(\alpha-1)\gamma}\right)
\end{equation}

The expected number of green tokens is the sum of these probabilities across all positions:

\begin{equation}
\mathbb{E}[|s|_G] = \sum_{t=1}^T P(s_t \in \mathcal{G}) \geq \frac{\gamma\alpha T}{1+(\alpha-1)\gamma} S^*
\end{equation}

Since each token selection is independent (conditioned on the context), the variance is the sum of the variances of Bernoulli random variables \cite{eaton1970note}:

\begin{equation}
\text{Var}[|s|_G] = \sum_{t=1}^T P(s_t \in \mathcal{G})(1-P(s_t \in \mathcal{G}))
\end{equation}

By Jensen's inequality \cite{mcshane1937jensen} and the concavity of the variance function $p(1-p)$, we obtain:

\begin{equation}
\text{Var}[|s|_G] \leq T \frac{\gamma\alpha S^*}{1+(\alpha-1)\gamma} \left(1-\frac{\gamma\alpha S^*}{1+(\alpha-1)\gamma}\right)
\end{equation}
\end{proof}

\subsection{Impact on Text Quality}
The impact of our watermarking approach on text quality can be bounded as follows:
\begin{theorem}
For a token at position $t$ with original probability distribution $p^{(t)}$ and watermarked distribution $\hat{p}^{(t)}$, the expected perplexity with respect to the randomness of the partition is bounded by:

\begin{equation}
\mathbb{E}_{G,R} \sum_k \hat{p}^{(t)}_k \ln(p^{(t)}_k) \leq (1+(\alpha-1)\gamma)P^*
\end{equation}
where $P^* = \sum_k p^{(t)}_k \ln(p^{(t)}_k)$ is the perplexity of the original model.
\end{theorem}

\begin{proof}
The probability of sampling token $k$ from the modified distribution is:

\begin{equation}
\hat{p}_k = \mathbb{E}_{G,R} \frac{\alpha^{\mathbb{I}(k \in G)} p_k}{\sum_{i \in R} p_i + \alpha \sum_{i \in G} p_i}
\end{equation}

We can decompose this expectation based on whether $k$ is in $G$ or $R$:

\begin{align}
\begin{split}
    \hat{p}_k &= \gamma \cdot \mathbb{E}_{G,R|k \in G} \frac{\alpha p_k}{\sum_{i \in R} p_i + \alpha \sum_{i \in G} p_i} \\ 
    &+ (1-\gamma) \cdot \mathbb{E}_{G,R|k \in R} \frac{p_k}{\sum_{i \in R} p_i + \alpha \sum_{i \in G} p_i} \\
&\leq \gamma \alpha p_k + (1-\gamma)p_k = (1+(\alpha-1)\gamma)p_k
\end{split}
\end{align}

The expected perplexity is then:

\begin{align}
\mathbb{E}_{G,R} \sum_k \hat{p}^{(t)}_k \ln(p^{(t)}_k) &= \sum_k \mathbb{E}_{G,R} \hat{p}^{(t)}_k \ln(p^{(t)}_k) \\
&\leq \sum_k (1+(\alpha-1)\gamma)p^{(t)}_k \ln(p^{(t)}_k) \\
&= (1+(\alpha-1)\gamma)P^*
\end{align}
\end{proof}

\subsection{Resilience Against Adversarial Attacks}

Let's analyze the resilience of our dual-layered approach against adversarial modifications. Suppose an adversary modifies $m$ tokens in a sequence of length $T$.

For the token distribution watermark, each modified token can affect at most $k+1$ tokens (itself and $k$ subsequent tokens whose red-green lists depend on it). The maximum reduction in the z-statistic is:

\begin{equation}
\Delta Z_{\text{token}} \leq \frac{(k+1)m}{\sqrt{T}}
\end{equation}

For the knowledge-based watermark, the adversary would need to identify and modify the specific watermarked knowledge, which requires knowledge of which facts serve as watermarks. If the adversary randomly modifies tokens, the probability of removing all watermarked knowledge is:

\begin{equation}
P(\text{remove all watermarks}) \leq \left(\frac{m}{T}\right)^{|W_D|}
\end{equation}
where $|W_D|$ is the number of knowledge-based watermarks.
This demonstrates that our dual-layered approach requires an adversary to make substantial modifications to the text to evade detection, which would significantly degrade the quality and utility of the generated content.

\section{Experiments}

\subsection{Experimental Setup}
To evaluate the effectiveness of our dual-layered watermarking approach for detecting unauthorized RAG usage, we conducted comprehensive experiments using our RPD dataset. This section details the experimental configuration and implementation specifics.
\subsubsection{RAG System Configuration}
We implemented a standard RAG pipeline consisting of an embedding-based retrieval system followed by a language model for content generation. The key components of our experimental setup include:

\begin{itemize}
    \item \textbf{Retrieval System:} We utilized OpenAI's \texttt{text-embedding-3-small}\footnote{https://platform.openai.com/docs/models/text-embedding-3-small} model to generate embedding representations of documents in our corpus. For efficient similarity search, we implemented an approximate nearest neighbor search using the \texttt{hnswlib} library \cite{malkov2018efficient}, which provides a balance between search speed and accuracy.
    
    \item \textbf{Retrieval Parameters:} For each query, we set the number of retrieved documents to $K=3$, which represents a typical configuration in production RAG systems that balances context richness with computational efficiency.
    
    \item \textbf{Large Language Models:} To ensure robustness across different generation capabilities, we employed multiple state-of-the-art language models for the generation component of our RAG system, including variants from the OpenAI, Gemini, and Claude families.
\end{itemize}

The rag prompt is shown in Figure \ref{figure_prompts_rag_detect}.

\subsubsection{Watermarking Implementation}
We implemented both layers of our watermarking approach with the following parameters:

\begin{itemize}
    \item \textbf{Knowledge-based Watermarking:} For embedding semantic-level watermarks, we utilized OpenAI's \texttt{text} \textit{-embedding-3-small} model to compute semantic similarities between Knowledge. This allowed us to identify candidate watermark Knowledge that maintain semantic coherence while being distinctive enough to serve as effective watermarks. The prompt of knowledge-based watermarking is shown in Figure \ref{figure_prompts_facts_watermark}.
    
    \item \textbf{Red-Green Token Distribution:} For the lexical-level watermarking, we set $\gamma=0.25$, meaning approximately 25\% of the vocabulary was designated as ``green'' tokens for any given context. We applied a bias strength of $\delta=4.0$ to increase the probability of selecting green tokens during text generation, creating a detectable statistical pattern while maintaining text fluency. The prompt of red-green watermarking is shown in the top of Figure \ref{figure_prompts_green_red}.
\end{itemize}

\subsubsection{Baselines}
To comprehensively evaluate our dual-layered watermarking approach, we compared it against several state-of-the-art methods for detecting unauthorized data usage in generative AI systems. These baselines represent different approaches to the problem, ranging from membership inference attacks to specialized watermarking techniques:
1.\textbf{AAG (Attacking Augmented Generation)} \cite{anderson2024my}: This method focuses on membership inference attacks against RAG systems. AAG creates carefully crafted prompts designed to determine whether specific text passages exist in a retrieval database by analyzing the RAG system's outputs. Originally developed to highlight privacy vulnerabilities in RAG systems, AAG can be repurposed to detect unauthorized data usage by treating it as a membership inference problem.
2.\textbf{Facts}: A strong baseline proposed by \cite{jovanovic2025ward}, which prompts an auxiliary LLM to generate a single question that is only answerable by reading a document. Then, this question will be sent to the suspicious RAG system. If the RAG system successfully answers the question, it indicates that data theft has occurred.
3.\textbf{WARD (Watermarking for RAG Dataset Inference)} \cite{jovanovic2025ward}: This method represents the current state-of-the-art approach specifically designed for RAG Dataset Inference (RAG-DI). WARD employs LLM watermarks to provide data owners with statistical guarantees regarding the misuse of their content in RAG corpora.

\subsubsection{Evaluation Metrics and Dataset Configuration}
To comprehensively evaluate our detection method, we established two distinct dataset configurations:

\begin{itemize}
    \item \textbf{D-IN (Dataset-Included):} RAG systems containing unauthorized watermarked documents in their retrieval corpus. These systems represent the positive cases where data theft has occurred.
    
    \item \textbf{D-OUT (Dataset-Excluded):} RAG systems without any watermarked documents in their retrieval corpus. These systems represent the negative cases where no data theft has occurred.
\end{itemize}

For each configuration, we report the following evaluation metrics:

\begin{itemize}
    \item \textbf{True Positives (TP):} Number of D-IN systems correctly identified as containing unauthorized data.
    \item \textbf{False Negatives (FN):} Number of D-IN systems incorrectly classified as not containing unauthorized data.
    \item \textbf{True Negatives (TN):} Number of D-OUT systems correctly identified as not containing unauthorized data.
    \item \textbf{False Positives (FP):} Number of D-OUT systems incorrectly classified as containing unauthorized data.
    \item \textbf{Accuracy (Acc.):} The proportion of correct classifications across both D-IN and D-OUT systems, calculated as $\text{Acc.} = \frac{\text{TP} + \text{TN}}{\text{TP} + \text{TN} + \text{FP} + \text{FN}}$.
\end{itemize}

We conducted experiments under two conditions: (1) without factual redundancy, where each fact appears in only one document, and (2) with factual redundancy, where similar facts appear across multiple documents with different expressions, simulating real-world information ecosystems.

\begin{figure*}[t]
	    \centering
	    \includegraphics[scale=0.41]{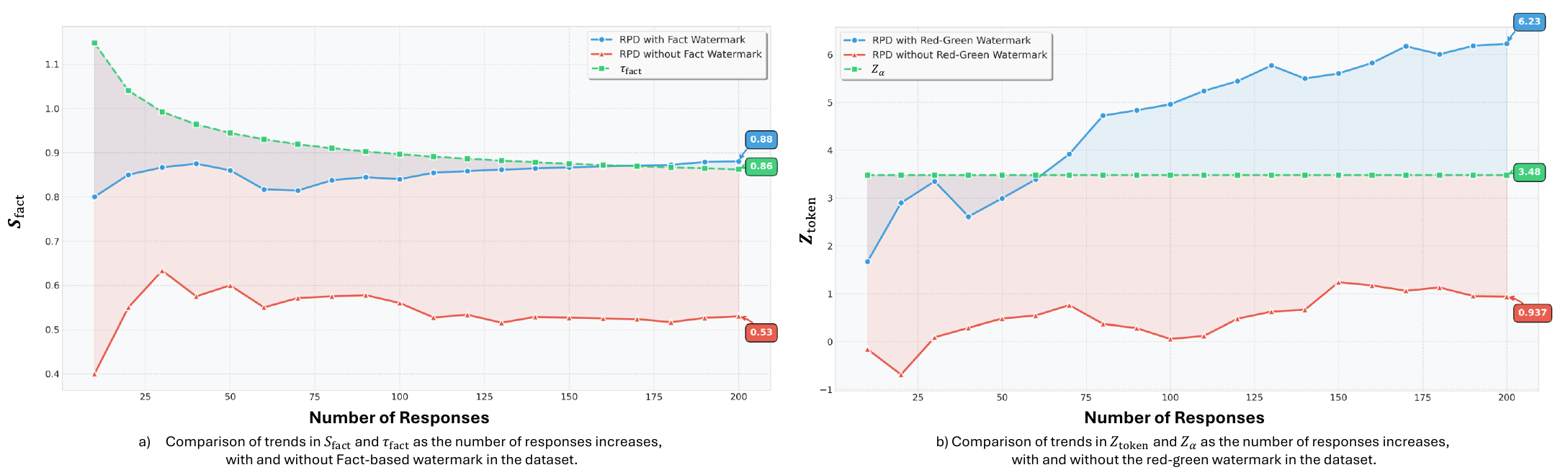}
	    \caption{Experimental results on the robustness detection of red-green watermark and Fact-based watermark. Using a z-test $(\alpha = 0.005)$.}
	    \label{figure_statical_experiments}
\end{figure*}

\subsection{Robustness of RAG Plagiarism Detection}

To validate the effectiveness of our dual-layered watermarking approach, we designed a comprehensive experimental framework that simulates real-world RAG scenarios. Our experiments specifically address a critical challenge in RAG plagiarism detection: while the watermark signal in any single piece of content may be too weak to reliably detect misuse, accumulated evidence across multiple queries can provide statistically significant detection power.
We split \textbf{RPD} into two distinct datasets to simulate authentic RAG environments:
\begin{itemize}
    \item \textbf{Unauthorized Dataset:} Content embedded with our dual-layered watermarks (both fact-based and red-green token distribution watermarks).
    \item \textbf{Authorized Dataset:} Similar content without watermarks.
\end{itemize}
These datasets were thoroughly intermixed to simulate the massive, heterogeneous knowledge bases typically used in production RAG systems. This mixed dataset was then indexed and made available to our experimental RAG system, allowing us to control the proportion of watermarked content while maintaining realistic retrieval dynamics.

Our evaluation focused on two key metrics that correspond to our dual-layered approach:
\begin{enumerate}
    \item \textbf{Fact-based Watermark Detection:} Measured by $S_{\text{fact}}$, the proportion of responses that are able to answer the query.
    \item \textbf{Red-Green Token Distribution Detection:} Measured by $Z_{\text{token}}$, the statistic of the standardized test for the bias of the token distribution.
\end{enumerate}
For each experiment, we deployed the Interrogator module to generate queries specifically designed to target watermarked content. These queries were constructed to be answerable only by accessing the watermarked information, forcing the RAG system to either retrieve the watermarked content (if available) or respond with ``Unanswerable'' when the information could not be found.

\subsubsection{Fact-based Watermark Detection Results}
We first evaluated the robustness of our fact-based watermarking approach by analyzing how detection performance scales with query volume. Figure \ref{figure_statical_experiments}a illustrates the relationship between the number of queries and the detection statistic $S_{\text{fact}}$.

When the RAG system had access to the mixed watermarked dataset, we observed that $S_{\text{fact}}$ increased steadily with the number of responses and eventually stabilized around $0.88$, significantly exceeding the detection threshold $\tau_{\text{fact}} = 0.86$ (calculated using $\alpha = 0.005$). This indicates strong statistical evidence of the use of unauthorized content.

In contrast, when the RAG system used only the non-watermarked dataset, $S_{\text{fact}}$ showed no increasing trend with query volume and stabilized around $0.53$, remaining consistently below the detection threshold $\tau_{\text{fact}} = 0.88$. This confirms that our detection method maintains a low false positive rate even with large query volumes.

Statistical significance was assessed using a z-test \cite{zaykin2011optimally} with $\alpha = 0.005$, demonstrating that the observed differences were highly unlikely to occur by chance. These results validate that our fact-based watermarking approach provides robust detection capabilities when sufficient queries are processed, despite the potential weakness of individual watermark signals.

\subsection{Red-Green Token Distribution Detection Results}
We next evaluated the robustness of our red-green token distribution watermarking approach. Figure \ref{figure_statical_experiments}b shows how the test statistic $Z_{\text{token}}$ varies with increasing query volume.

When the RAG system used the mixed watermark dataset, $Z_{\text{token}}$ exhibited a clear increasing trend with the number of responses, eventually stabilizing around $6.23$, significantly exceeding the critical value $z_{\alpha} = 3.48$ (corresponding to $\alpha = 0.005$). This strong statistical signal confirms that the token distribution bias introduced by our watermarking technique persists even after LLM reformulation in the RAG process.

In contrast, when the RAG system accessed only the non-watermarked dataset, $Z_{\text{token}}$ did not show a consistent increasing trend, fluctuating around $0.78$ and remaining well below the critical value $z_{\alpha} = 3.48$. This demonstrates that our token distribution watermarking maintains a low false-positive rate, correctly identifying systems that do not use watermarked content.

\begin{table*}[ht]
\centering
\resizebox{\textwidth}{!}{%
\begin{tabular}{llcccccccccc}
\hline
\multirow{3}{*}{\textbf{Method}} & \multirow{3}{*}{\textbf{LLM}} & \multicolumn{5}{c}{\textbf{Without Factual Redundancy}} & \multicolumn{5}{c}{\textbf{With Factual Redundancy}} \\
\cmidrule(lr){3-7} \cmidrule(lr){8-12}
& & \multicolumn{2}{c}{\textbf{D-IN}} & \multicolumn{2}{c}{\textbf{D-OUT}} & \multirow{2}{*}{\textbf{Acc.}\small{$\uparrow$}} & \multicolumn{2}{c}{\textbf{D-IN}} & \multicolumn{2}{c}{\textbf{D-OUT}} & \multirow{2}{*}{\textbf{Acc.}\small{$\uparrow$}} \\
\cmidrule(lr){3-4} \cmidrule(lr){5-6} \cmidrule(lr){8-9} \cmidrule(lr){10-11}
& & \textbf{TP}\small{$\uparrow$} & \textbf{FN}\small{$\downarrow$} & \textbf{TN}\small{$\uparrow$} & \textbf{FP}\small{$\downarrow$} & & \textbf{TP}\small{$\uparrow$} & \textbf{FN}\small{$\downarrow$} & \textbf{TN}\small{$\uparrow$} & \textbf{FP}\small{$\downarrow$} & \\
\hline
\multirow{4}{*}{AAG \cite{anderson2024my}} 
& GPT-4.1 & 10/10 & 0/10 & 10/10 & 0/10 & 1.0 & 10/10 & 0/10 & \cellcolor{cRed}0/10 & \cellcolor{cRed}10/10 & \cellcolor{cRed}0.5 \\
& OpenAI o3 & 10/10 & 0/10 & 10/10 & 0/10 & 1.0 & 10/10 & 0/10 & \cellcolor{cRed}0/10 & \cellcolor{cRed}10/10 & \cellcolor{cRed}0.5 \\
& Claude 3.5 & 10/10 & 0/10 & 10/10 & 0/10 & 1.0 & 10/10 & 0/10 & \cellcolor{cRed}0/10 & \cellcolor{cRed}10/10 & \cellcolor{cRed}0.5 \\
& Gemini-2.0 & 10/10 & 0/10 & 10/10 & 0/10 & 1.0 & 10/10 & 0/10 & \cellcolor{cRed}0/10 & \cellcolor{cRed}10/10 & \cellcolor{cRed}0.5 \\
\hline
\multirow{4}{*}{Facts \cite{jovanovic2025ward}} 
& GPT-4.1 & 10/10 & 0/10 & 10/10 & 0/10 & 1.0 & 10/10 & 0/10 & \cellcolor{cRed}0/10 & \cellcolor{cRed}10/10 & \cellcolor{cRed}0.5 \\
& OpenAI o3 & 10/10 & 0/10 & 10/10 & 0/10 & 1.0 & 10/10 & 0/10 & \cellcolor{cRed}0/10 & \cellcolor{cRed}10/10 & \cellcolor{cRed}0.5 \\
& Claude 3.5 & 10/10 & 0/10 & 10/10 & 0/10 & 1.0 & 10/10 & 0/10 & \cellcolor{cRed}0/10 & \cellcolor{cRed}10/10 & \cellcolor{cRed}0.5 \\
& Gemini-2.0 & 10/10 & 0/10 & 10/10 & 0/10 & 1.0 & 10/10 & 0/10 & \cellcolor{cRed}0/10 & \cellcolor{cRed}10/10 & \cellcolor{cRed}0.5 \\
\hline
\multirow{4}{*}{WARD \cite{jovanovic2025ward}} 
& GPT-4.1 & \cellcolor{cGrey}\textbf{10/10} & \cellcolor{cGrey}\textbf{0/10} & \cellcolor{cGrey}\textbf{10/10} & \cellcolor{cGrey}\textbf{0/10} & \cellcolor{cGrey}\textbf{1.0} & \cellcolor{cGrey}\textbf{10/10} & \cellcolor{cGrey}\textbf{0/10} & \cellcolor{cGrey}\textbf{10/10} & \cellcolor{cGrey}\textbf{0/10} & \cellcolor{cGrey}\textbf{1.0} \\
& OpenAI o3 & \cellcolor{cGrey}\textbf{10/10} & \cellcolor{cGrey}\textbf{0/10} & \cellcolor{cGrey}\textbf{10/10} & \cellcolor{cGrey}\textbf{0/10} & \cellcolor{cGrey}\textbf{1.0} & \cellcolor{cGrey}\textbf{10/10} & \cellcolor{cGrey}\textbf{0/10} & \cellcolor{cGrey}\textbf{10/10} & \cellcolor{cGrey}\textbf{0/10} & \cellcolor{cGrey}\textbf{1.0} \\
& Claude 3.5 & \cellcolor{cGrey}\textbf{10/10} & \cellcolor{cGrey}\textbf{0/10} & \cellcolor{cGrey}\textbf{10/10} & \cellcolor{cGrey}\textbf{0/10} & \cellcolor{cGrey}\textbf{1.0} & \cellcolor{cGrey}\textbf{10/10} & \cellcolor{cGrey}\textbf{0/10} & \cellcolor{cGrey}\textbf{10/10} & \cellcolor{cGrey}\textbf{0/10} & \cellcolor{cGrey}\textbf{1.0} \\
& Gemini-2.0 & \cellcolor{cGrey}\textbf{10/10} & \cellcolor{cGrey}\textbf{0/10} & \cellcolor{cGrey}\textbf{10/10} & \cellcolor{cGrey}\textbf{0/10} & \cellcolor{cGrey}\textbf{1.0} & \cellcolor{cGrey}\textbf{10/10} & \cellcolor{cGrey}\textbf{0/10} & \cellcolor{cGrey}\textbf{10/10} & \cellcolor{cGrey}\textbf{0/10} & \cellcolor{cGrey}\textbf{1.0} \\
\hline
\multirow{4}{*}{\textbf{Dual-Layered (Ours)}} 
& GPT-4.1 & \cellcolor{cGrey}\textbf{10/10} & \cellcolor{cGrey}\textbf{0/10} & \cellcolor{cGrey}\textbf{10/10} & \cellcolor{cGrey}\textbf{0/10} & \cellcolor{cGrey}\textbf{1.0} & \cellcolor{cGrey}\textbf{10/10} & \cellcolor{cGrey}\textbf{0/10} & \cellcolor{cGrey}\textbf{10/10} & \cellcolor{cGrey}\textbf{0/10} & \cellcolor{cGrey}\textbf{1.0} \\
& OpenAI o3 & \cellcolor{cGrey}\textbf{10/10} & \cellcolor{cGrey}\textbf{0/10} & \cellcolor{cGrey}\textbf{10/10} & \cellcolor{cGrey}\textbf{0/10} & \cellcolor{cGrey}\textbf{1.0} & \cellcolor{cGrey}\textbf{10/10} & \cellcolor{cGrey}\textbf{0/10} & \cellcolor{cGrey}\textbf{10/10} & \cellcolor{cGrey}\textbf{0/10} & \cellcolor{cGrey}\textbf{1.0} \\
& Claude 3.5 & \cellcolor{cGrey}\textbf{10/10} & \cellcolor{cGrey}\textbf{0/10} & \cellcolor{cGrey}\textbf{10/10} & \cellcolor{cGrey}\textbf{0/10} & \cellcolor{cGrey}\textbf{1.0} & \cellcolor{cGrey}\textbf{10/10} & \cellcolor{cGrey}\textbf{0/10} & \cellcolor{cGrey}\textbf{10/10} & \cellcolor{cGrey}\textbf{0/10} & \cellcolor{cGrey}\textbf{1.0} \\
& Gemini-2.0 & \cellcolor{cGrey}\textbf{10/10} & \cellcolor{cGrey}\textbf{0/10} & \cellcolor{cGrey}\textbf{10/10} & \cellcolor{cGrey}\textbf{0/10} & \cellcolor{cGrey}\textbf{1.0} & \cellcolor{cGrey}\textbf{10/10} & \cellcolor{cGrey}\textbf{0/10} & \cellcolor{cGrey}\textbf{10/10} & \cellcolor{cGrey}\textbf{0/10} & \cellcolor{cGrey}\textbf{1.0} \\
\hline
\end{tabular}%
}
\caption{Comprehensive detection performance across LLMs and dataset conditions}
\label{tab:comprehensive_results}
\end{table*}

\subsection{Comparison Results}

To evaluate the effectiveness of our dual-layered watermarking approach against existing methods, we conducted comprehensive experiments using the RPD dataset. We constructed an indexed dataset from RPD, watermarked $200$ documents to simulate unauthorized data, and deployed the Interrogator to generate targeted queries based on the watermarked content. These queries were then sent to RAG systems built with four widely-used LLMs: GPT-4.1, OpenAI o3, Claude 3.5 Sonnet, and Gemini-2.0. Each experiment was repeated $10$ times with different random seeds to ensure statistical robustness.
We list all prompts, and examples of RPD samples and watermarked documents, in Appendix \ref{app:prompts}.

\subsubsection{Detection Performance without Factual Redundancy}

As shown in Table \ref{tab:comprehensive_results}, all methods performed exceptionally well on datasets without factual redundancy, achieving perfect detection accuracy of $100\%$ across all LLMs. This confirms our hypothesis that in the absence of factual redundancy, detection is relatively straightforward since the queries constructed by the Interrogator can precisely match the watermarked documents without interference from similar content. In this scenario, both baseline methods and our approach successfully identified all instances of unauthorized data usage while correctly classifying systems without watermarked content.

\subsubsection{Detection Performance with Factual Redundancy}

The more challenging and realistic scenario involves datasets with factual redundancy, where multiple documents contain similar or identical core facts expressed in different writing styles. This setting more accurately reflects real-world information ecosystems and poses a significant challenge for detection methods. The results in Table \ref{tab:comprehensive_results} reveal substantial differences in performance among the methods under these conditions:

\begin{itemize}
    \item \textbf{AAG \cite{anderson2024my}:} Performance degraded substantially, with accuracy dropping to $50\%$ across all LLMs. Further analysis revealed that AAG consistently misclassified RAG systems that did not contain watermarked data (D-OUT) as containing unauthorized content, resulting in a $100\%$ false positive rate. This occurs because AAG's core detection mechanism relies on directly querying the LLM about document inclusion in the RAG corpus. When factually similar documents exist, the LLM can be misled into incorrectly confirming the presence of specific content.
    
    \item \textbf{Facts \cite{jovanovic2025ward}:} Similarly, this method achieved only $50\%$ accuracy with factual redundancy. While it correctly identified all D-IN systems (perfect true positive rate), it failed to correctly classify any D-OUT systems, resulting in a $100\%$ false positive rate. This is because when factual redundancy exists in the RAG dataset, even without watermarked documents, the LLM can obtain answers from other documents that contain similar information, leading to incorrect classification.
    
    \item \textbf{WARD \cite{jovanovic2025ward}:} This method maintained strong performance even with factual redundancy, achieving $100\%$ accuracy across all LLMs. Its effectiveness stems from the statistical token distribution patterns that remain detectable even when similar content is retrieved, providing resilience against factual redundancy challenges.
    
    \item \textbf{Dual-Layered (Ours):} Our approach also achieved perfect $100\%$ accuracy across all LLMs in the presence of factual redundancy. This robust performance can be attributed to the synergistic effect of combining knowledge-based watermarks with token distribution manipulation, creating distinctive signatures that stand out even among similar information.
\end{itemize}

These results demonstrate that while simple fact-based or membership inference approaches may work well in controlled environments without factual redundancy, they fail to maintain reliability in more realistic scenarios. In contrast, both WARD and our dual-layered approach provide robust detection capabilities that remain effective even in the presence of factual redundancy, with our method offering the additional benefit of a multi-layered protection mechanism that can withstand various evasion attempts.

\begin{table*}[t]
\centering
\resizebox{\textwidth}{!}{%
\begin{tabular}{llcccccccccc}
\hline
\multirow{3}{*}{\textbf{Method}} & \multirow{3}{*}{\textbf{LLM}} & \multicolumn{5}{c}{\textbf{Knowledge Watermark Evasion}} & \multicolumn{5}{c}{\textbf{Token Distribution Evasion}} \\
\cmidrule(lr){3-7} \cmidrule(lr){8-12}
& & \multicolumn{2}{c}{\textbf{D-IN}} & \multicolumn{2}{c}{\textbf{D-OUT}} & \multirow{2}{*}{\textbf{Acc.}\small{$\uparrow$}} & \multicolumn{2}{c}{\textbf{D-IN}} & \multicolumn{2}{c}{\textbf{D-OUT}} & \multirow{2}{*}{\textbf{Acc.}\small{$\uparrow$}} \\
\cmidrule(lr){3-4} \cmidrule(lr){5-6} \cmidrule(lr){8-9} \cmidrule(lr){10-11}
& & \textbf{TP}\small{$\uparrow$} & \textbf{FN}\small{$\downarrow$} & \textbf{TN}\small{$\uparrow$} & \textbf{FP}\small{$\downarrow$} & & \textbf{TP}\small{$\uparrow$} & \textbf{FN}\small{$\downarrow$} & \textbf{TN}\small{$\uparrow$} & \textbf{FP}\small{$\downarrow$} & \\
\hline
\multirow{4}{*}{WARD \cite{jovanovic2025ward}}
& GPT-4.1 & 10/10 & 0/10 & 10/10 & 0/10 & 1.0 & \cellcolor{cRed}0/10 & \cellcolor{cRed}10/10 & 10/10 & 0/10 & \cellcolor{cRed}0.5 \\
& OpenAI o3 & 10/10 & 0/10 & 10/10 & 0/10 & 1.0 & \cellcolor{cRed}0/10 & \cellcolor{cRed}10/10 & 10/10 & 0/10 & \cellcolor{cRed}0.5 \\
& Claude 3.5 & 10/10 & 0/10 & 10/10 & 0/10 & 1.0 & 4/10 & 6/10 & 10/10 & 0/10 & 0.7 \\
& Gemini-2.0 & 10/10 & 0/10 & 10/10 & 0/10 & 1.0 & \cellcolor{cRed}0/10 & \cellcolor{cRed}10/10 & 10/10 & 0/10 & \cellcolor{cRed}0.5 \\
\hline
\multirow{4}{*}{\textbf{Dual-Layered (Ours)}}
& GPT-4.1 & \cellcolor{cGrey}\textbf{10/10} & \cellcolor{cGrey}\textbf{0/10} & \cellcolor{cGrey}\textbf{10/10} & \cellcolor{cGrey}\textbf{0/10} & \cellcolor{cGrey}\textbf{1.0} & \cellcolor{cGrey}\textbf{10/10} & \cellcolor{cGrey}\textbf{0/10} & \cellcolor{cGrey}\textbf{10/10} & \cellcolor{cGrey}\textbf{0/10} & \cellcolor{cGrey}\textbf{1.0}  \\
& OpenAI o3 & \cellcolor{cGrey}\textbf{10/10} & \cellcolor{cGrey}\textbf{0/10} & \cellcolor{cGrey}\textbf{10/10} & \cellcolor{cGrey}\textbf{0/10} & \cellcolor{cGrey}\textbf{1.0} & \cellcolor{cGrey}\textbf{10/10} & \cellcolor{cGrey}\textbf{0/10} & \cellcolor{cGrey}\textbf{10/10} & \cellcolor{cGrey}\textbf{0/10} & \cellcolor{cGrey}\textbf{1.0}  \\
& Claude 3.5 & \cellcolor{cGrey}\textbf{10/10} & \cellcolor{cGrey}\textbf{0/10} & \cellcolor{cGrey}\textbf{10/10} & \cellcolor{cGrey}\textbf{0/10} & \cellcolor{cGrey}\textbf{1.0} & \cellcolor{cGrey}\textbf{10/10} & \cellcolor{cGrey}\textbf{0/10} & \cellcolor{cGrey}\textbf{10/10} & \cellcolor{cGrey}\textbf{0/10} & \cellcolor{cGrey}\textbf{1.0}  \\
& Gemini-2.0 & \cellcolor{cGrey}\textbf{10/10} & \cellcolor{cGrey}\textbf{0/10} & \cellcolor{cGrey}\textbf{10/10} & \cellcolor{cGrey}\textbf{0/10} & \cellcolor{cGrey}\textbf{1.0} & \cellcolor{cGrey}\textbf{10/10} & \cellcolor{cGrey}\textbf{0/10} & \cellcolor{cGrey}\textbf{10/10} & \cellcolor{cGrey}\textbf{0/10} & \cellcolor{cGrey}\textbf{1.0}  \\
\hline
\end{tabular}%
}
\caption{Comprehensive detection performance under adversarial conditions}
\label{tab:adversarial_results}
\end{table*}

\subsection{Adversarial Experiments}
To evaluate the robustness of our approach against sophisticated evasion attempts, we conducted adversarial experiments simulating realistic scenarios where RAG data thieves actively attempt to conceal their unauthorized data usage. Since our previous experiments demonstrated that AAG and Facts methods fail to handle factual redundancy (achieving only 50\% accuracy), we focused our adversarial analysis on comparing our Dual-Layered approach against WARD \cite{jovanovic2025ward}, the current state-of-the-art method.
We implemented two common adversarial tactics employed by RAG data thieves:
\begin{itemize}
    \item \textbf{Knowledge Watermark Evasion:} This technique targets knowledge-based watermarks by using LLMs to detect incoherently spliced facts within documents. When potentially watermarked documents are identified, the RAG system is instructed to avoid referencing these suspicious documents, thereby attempting to evade detection. The prompt is shown in Figure \ref{figure_prompts_facts_adversarial}.
    
    \item \textbf{Token Distribution Evasion:} This approach targets token-distribution watermarks by employing LLMs to completely rewrite referenced articles, minimizing the use of original words and phrases to disrupt statistical patterns embedded in the text. The prompt is shown in Figure \ref{figure_prompts_adversarial}.
\end{itemize}
These evasion techniques represent realistic adversarial scenarios that any practical watermarking system must withstand to provide reliable protection against unauthorized data usage.
Table \ref{tab:adversarial_results} presents the performance of both methods under adversarial conditions. The results reveal significant differences in robustness:
\begin{itemize}
    \item \textbf{Knowledge Watermark Evasion:} Both WARD and our Dual-Layered approach maintained perfect detection accuracy (100\%) across all LLMs when facing knowledge watermark evasion attempts. This suggests that token distribution-based watermarking techniques remain effective even when adversaries attempt to identify and avoid watermarked documents.
    
    \item \textbf{Token Distribution Evasion:} WARD's performance degraded catastrophically when facing token distribution evasion, with accuracy dropping to 50\% for GPT-4.1, OpenAI o3, and Gemini 2.0, and 70\% for Claude 3.5. The true positive rate (TP) dropped to 0/10 for three of the four LLMs, indicating complete failure to detect unauthorized data usage when adversaries employ rewriting techniques. In contrast, our Dual-Layered approach maintained perfect 100\% accuracy across all LLMs, demonstrating exceptional resilience against token distribution evasion attempts.
\end{itemize}
The stark difference in performance under token distribution evasion can be attributed to the fundamental design principles of each approach:
\begin{itemize}
    \item WARD relies primarily on statistical patterns in token distributions, which can be effectively disrupted through comprehensive rewriting. When adversaries employ LLMs to reformulate content while preserving core information, the statistical watermarks become undetectable.
    
    \item Our Dual-Layered approach achieves superior robustness through its sequential integration of complementary watermarking techniques. By combining knowledge-based watermarks with token distribution manipulation, our method creates redundant protection mechanisms that remain effective even when one layer is compromised:
    
    \begin{itemize}
        \item The knowledge-based watermarks are designed with careful semantic coherence, making them difficult to detect and filter out without compromising information retrieval quality.
        
        \item Our Interrogator module generates queries specifically targeting the unique knowledge combinations embedded in watermarked documents. Even if adversaries completely rewrite the documents, they cannot answer these queries without referencing the watermarked content.
        
        \item The dual-layer design ensures that even if token distribution patterns are disrupted through rewriting, the knowledge-based watermarks remain intact and detectable.
    \end{itemize}
\end{itemize}

\begin{table*}[t]
\centering
\resizebox{\textwidth}{!}{%
\begin{tabular}{llcccccccccc}
\hline
\multirow{3}{*}{\textbf{Method}} & \multirow{3}{*}{\textbf{LLM}} & \multicolumn{5}{c}{\textbf{Knowledge Watermark Evasion}} & \multicolumn{5}{c}{\textbf{Token Distribution Evasion}} \\
\cmidrule(lr){3-7} \cmidrule(lr){8-12}
& & \multicolumn{2}{c}{\textbf{D-IN}} & \multicolumn{2}{c}{\textbf{D-OUT}} & \multirow{2}{*}{\textbf{Acc.}\small{$\uparrow$}} & \multicolumn{2}{c}{\textbf{D-IN}} & \multicolumn{2}{c}{\textbf{D-OUT}} & \multirow{2}{*}{\textbf{Acc.}\small{$\uparrow$}} \\
\cmidrule(lr){3-4} \cmidrule(lr){5-6} \cmidrule(lr){8-9} \cmidrule(lr){10-11}
& & \textbf{TP}\small{$\uparrow$} & \textbf{FN}\small{$\downarrow$} & \textbf{TN}\small{$\uparrow$} & \textbf{FP}\small{$\downarrow$} & & \textbf{TP}\small{$\uparrow$} & \textbf{FN}\small{$\downarrow$} & \textbf{TN}\small{$\uparrow$} & \textbf{FP}\small{$\downarrow$} & \\
\hline
\multirow{4}{*}{KW-Only}
& GPT-4.1 & \cellcolor{cRed}4/10 & \cellcolor{cRed}6/10 & 10/10 & 0/10 & \cellcolor{cRed}0.7 & 10/10 & 0/10 & 10/10 & 0/10 & 1.0 \\
& OpenAI o3 & \cellcolor{cRed}6/10 & \cellcolor{cRed}4/10 & 10/10 & 0/10 & \cellcolor{cRed}0.8 & 10/10 & 0/10 & 10/10 & 0/10 & 1.0 \\
& Claude 3.5 & \cellcolor{cRed}8/10 & \cellcolor{cRed}2/10 & 10/10 & 0/10 & \cellcolor{cRed}0.9 & 10/10 & 0/10 & 10/10 & 0/10 & 1.0 \\
& Gemini-2.0 & \cellcolor{cRed}4/10 & \cellcolor{cRed}6/10 & 10/10 & 0/10 & \cellcolor{cRed}0.7 & 10/10 & 0/10 & 10/10 & 0/10 & 1.0 \\
\hline
\multirow{4}{*}{TD-Only}
& GPT-4.1 & 10/10 & 0/10 & 10/10 & 0/10 & 1.0 & \cellcolor{cRed}0/10 & \cellcolor{cRed}10/10 & 10/10 & 0/10 & \cellcolor{cRed}0.5 \\
& OpenAI o3 & 10/10 & 0/10 & 10/10 & 0/10 & 1.0 & \cellcolor{cRed}0/10 & \cellcolor{cRed}10/10 & 10/10 & 0/10 & \cellcolor{cRed}0.5 \\
& Claude 3.5 & 10/10 & 0/10 & 10/10 & 0/10 & 1.0 & 4/10 & 6/10 & 10/10 & 0/10 & 0.7 \\
& Gemini-2.0 & 10/10 & 0/10 & 10/10 & 0/10 & 1.0 & \cellcolor{cRed}0/10 & \cellcolor{cRed}10/10 & 10/10 & 0/10 & \cellcolor{cRed}0.5 \\
\hline
\multirow{4}{*}{\textbf{Dual-Layered (Ours)}}
& GPT-4.1 & \cellcolor{cGrey}\textbf{10/10} & \cellcolor{cGrey}\textbf{0/10} & \cellcolor{cGrey}\textbf{10/10} & \cellcolor{cGrey}\textbf{0/10} & \cellcolor{cGrey}\textbf{1.0} & \cellcolor{cGrey}\textbf{10/10} & \cellcolor{cGrey}\textbf{0/10} & \cellcolor{cGrey}\textbf{10/10} & \cellcolor{cGrey}\textbf{0/10} & \cellcolor{cGrey}\textbf{1.0}  \\
& OpenAI o3 & \cellcolor{cGrey}\textbf{10/10} & \cellcolor{cGrey}\textbf{0/10} & \cellcolor{cGrey}\textbf{10/10} & \cellcolor{cGrey}\textbf{0/10} & \cellcolor{cGrey}\textbf{1.0} & \cellcolor{cGrey}\textbf{10/10} & \cellcolor{cGrey}\textbf{0/10} & \cellcolor{cGrey}\textbf{10/10} & \cellcolor{cGrey}\textbf{0/10} & \cellcolor{cGrey}\textbf{1.0}  \\
& Claude 3.5 & \cellcolor{cGrey}\textbf{10/10} & \cellcolor{cGrey}\textbf{0/10} & \cellcolor{cGrey}\textbf{10/10} & \cellcolor{cGrey}\textbf{0/10} & \cellcolor{cGrey}\textbf{1.0} & \cellcolor{cGrey}\textbf{10/10} & \cellcolor{cGrey}\textbf{0/10} & \cellcolor{cGrey}\textbf{10/10} & \cellcolor{cGrey}\textbf{0/10} & \cellcolor{cGrey}\textbf{1.0}  \\
& Gemini-2.0 & \cellcolor{cGrey}\textbf{10/10} & \cellcolor{cGrey}\textbf{0/10} & \cellcolor{cGrey}\textbf{10/10} & \cellcolor{cGrey}\textbf{0/10} & \cellcolor{cGrey}\textbf{1.0} & \cellcolor{cGrey}\textbf{10/10} & \cellcolor{cGrey}\textbf{0/10} & \cellcolor{cGrey}\textbf{10/10} & \cellcolor{cGrey}\textbf{0/10} & \cellcolor{cGrey}\textbf{1.0}  \\
\hline
\end{tabular}%
}
\caption{Ablation study of Knowledge-based watermarking and Red-Green Token Distribution Manipulation.}
\label{tab:ablation_results}
\end{table*}

\subsection{Ablation Study}
To understand the individual contributions of each watermarking layer and validate the synergistic benefits of our dual-layered approach, we conducted a comprehensive ablation study. This analysis isolated the performance of knowledge-based watermarking (KW-Only) and token distribution manipulation (TD-Only) against the full dual-layered implementation across various adversarial scenarios.
Table \ref{tab:ablation_results} presents the results of this ablation study, revealing critical insights into the complementary nature of our watermarking techniques:

\subsubsection{Knowledge-Based Watermarking in Isolation}
When knowledge-based watermarking was employed as a standalone technique (KW-Only), we observed:
\begin{itemize}
    \item \textbf{Vulnerability to Knowledge Watermark Evasion:} Performance degraded significantly when facing targeted evasion attempts, with accuracy dropping to 70-90\% across different LLMs. GPT-4.1 and Gemini-2.0 showed particular vulnerability with only 4/10 true positives, while Claude 3.5 demonstrated greater resilience with 8/10 true positives.
    \item \textbf{Resilience Against Token Distribution Evasion:} Interestingly, knowledge-based watermarking maintained perfect 100\% accuracy when facing token distribution evasion attempts. This confirms our hypothesis that semantic-level watermarks remain detectable even when adversaries completely rewrite the text, as long as the core information is preserved.
\end{itemize}
This pattern illustrates a fundamental characteristic of knowledge-based watermarking: its effectiveness depends on the preservation of specific factual information rather than lexical or syntactic patterns. When adversaries specifically target knowledge watermarks by filtering out suspicious documents, detection capability suffers. However, when they focus on rewriting content while preserving core information, knowledge watermarks remain robust.

\subsubsection{Token Distribution Manipulation in Isolation}
The token distribution approach (TD-Only) exhibited the inverse pattern:
\begin{itemize}
    \item \textbf{Resilience Against Knowledge Watermark Evasion:} Token distribution watermarking achieved perfect 100\% accuracy across all LLMs when facing knowledge watermark evasion. This demonstrates that statistical patterns in token distributions remain detectable even when adversaries attempt to identify and filter out watermarked documents.
    \item \textbf{Vulnerability to Token Distribution Evasion:} Performance collapsed dramatically when facing targeted rewriting, with accuracy dropping to 50\% for most LLMs (70\% for Claude 3.5). The true positive rate fell to 0/10 for three of the four LLMs, indicating complete failure to detect unauthorized data usage when adversaries employ comprehensive rewriting techniques.
\end{itemize}
This pattern reveals the complementary weakness of token distribution watermarking: while it creates robust statistical signatures that are difficult to identify through document filtering, these signatures can be effectively disrupted through comprehensive rewriting that preserves semantic content while altering lexical choices.

\subsubsection{The Synergistic Effect of Dual-Layered Watermarking}
The full dual-layered approach demonstrated perfect 100\% accuracy across all experimental conditions, highlighting the synergistic protection achieved through sequential integration:
\begin{itemize}
    \item \textbf{Complementary Protection:} Each layer effectively compensates for the vulnerabilities of the other. When adversaries target knowledge watermarks through filtering, the token distribution layer maintains detection capability. Conversely, when adversaries employ rewriting to disrupt token distributions, the knowledge-based layer preserves detection capability.
    \item \textbf{Increased Adversarial Cost:} To evade detection, adversaries would need to simultaneously implement both evasion strategies—filtering out watermarked documents and comprehensively rewriting any remaining content that might contain watermarks. This substantially increases the computational and quality costs of evasion attempts.
    \item \textbf{Cross-Layer Reinforcement:} The sequential application of both watermarking techniques creates a reinforced protection mechanism where the knowledge-based watermarks guide the Interrogator to generate queries that specifically target watermarked content, while the token distribution layer provides statistical evidence of unauthorized usage even when content is partially reformulated.
\end{itemize}
This ablation study empirically validates our theoretical framework for dual-layered watermarking. The results demonstrate that while each individual watermarking technique has specific vulnerabilities to targeted evasion attempts, their sequential integration creates a robust detection system that maintains effectiveness across diverse adversarial scenarios. This synergistic effect is particularly valuable in real-world applications where RAG system operators may employ various sophisticated techniques to conceal unauthorized data usage.
The complementary nature of these watermarking approaches—one operating at the semantic level and the other at the lexical level—creates a protection mechanism that addresses the fundamental tension in watermarking: creating marks that are simultaneously imperceptible to casual readers yet reliably detectable by authorized verification systems, even in the face of determined adversarial efforts.

\begin{figure*}[t]
	    \centering
	    \includegraphics[scale=0.5]{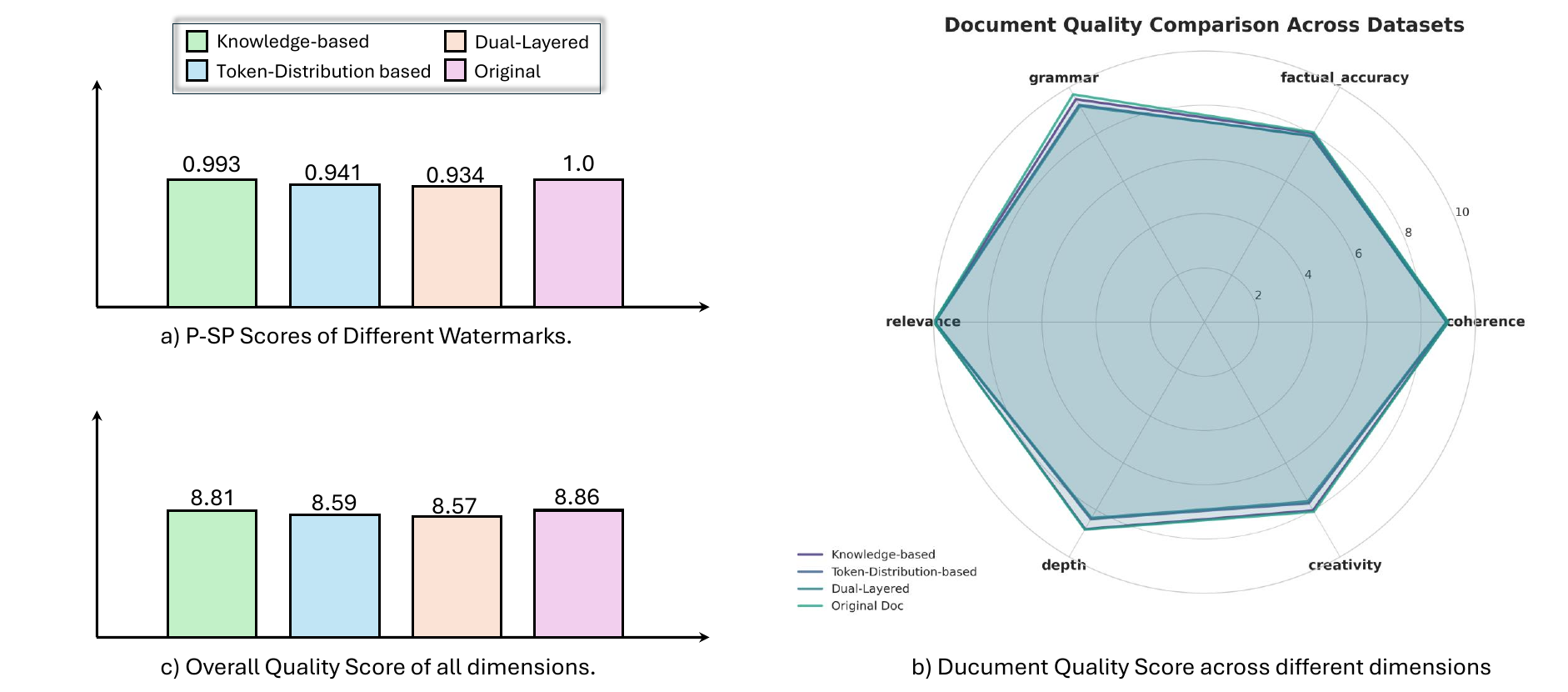}
	    \caption{Text quality assessment across different watermarking approaches.}
	    \label{figure_text_quality}
\end{figure*}

\subsection{Impact on Text Quality}
A critical consideration for any watermarking technique is its impact on the quality of the watermarked text. Effective watermarks must be imperceptible to human readers while remaining detectable by verification systems. To evaluate this aspect of our approach, we conducted a comprehensive quality assessment comparing original documents against their watermarked counterparts.

\subsubsection{Evaluation Methodology}
We employed two complementary evaluation approaches to assess the quality impact of our watermarking techniques:

\begin{itemize}
    \item \textbf{LLM-based Evaluation:} Following established practices in natural language generation evaluation \cite{fu2023gptscore, liu2023g}, we utilized GPT-4.1 as an evaluator to assess text quality across multiple dimensions. This approach provides a multi-faceted analysis of text quality from a human-like perspective. Detailed prompts can be found in the bottom of Figure \ref{figure_prompts_green_red}.
    
    \item \textbf{Embedding-based Evaluation:} We employed the P-SP metric \cite{DBLP:conf/emnlp/WietingGNB22}, which measures semantic preservation through embedding similarity. This provides an objective, quantitative measure of how closely watermarked text preserves the semantic content of the original.
\end{itemize}

For each document in our test set, we created three variants: (1) with knowledge-based watermarking only, (2) with token distribution watermarking only, and (3) with our full dual-layered watermarking approach. These variants, along with the original unwatermarked document, were evaluated on seven quality dimensions using GPT-4.1:

\begin{itemize}
    \item \textbf{Coherence:} Logical flow and connection between ideas
    \item \textbf{Factual Accuracy:} Correctness of presented information
    \item \textbf{Grammar:} Syntactic correctness and linguistic fluency
    \item \textbf{Relevance:} Alignment with the document's intended topic
    \item \textbf{Depth:} Thoroughness and comprehensiveness of content
    \item \textbf{Creativity:} Originality and engaging presentation
    \item \textbf{Overall Score:} Composite assessment of document quality
\end{itemize}
Each dimension was scored on a scale from $1$ to $10$, with higher scores indicating better quality. For the P-SP metric, scores range from $0$ to $1$, with higher values indicating greater semantic preservation relative to the original document.

\subsubsection{Quality Assessment Results}
Figure \ref{figure_text_quality} presents the results of our quality assessment, revealing several important insights about the impact of different watermarking approaches on text quality.

\textbf{Knowledge-based Watermarking:} This approach demonstrated the smallest quality degradation, with an overall GPT-4.1 score of $8.81$ compared to $8.87$ for original documents—a minimal reduction of only 0.06 points. The P-SP score of $0.993$ further confirms excellent semantic preservation, with only a $0.7\%$ reduction from the original. The preservation of quality can be attributed to the semantic integration of watermark facts, which maintains document coherence ($8.99$) and relevance ($9.99$). The largest impact was observed in grammar ($9.49$ vs. $9.70$), likely due to the insertion of additional factual statements that may occasionally disrupt the original syntactic flow.

\textbf{Token-Distribution Watermarking:} This technique showed more noticeable quality impacts, with an overall GPT-4.1 score of $8.60$, a reduction of $0.27$ points from the original. The P-SP score dropped to $0.941$, indicating a more substantial semantic shift. The most significant degradations occurred in creativity ($7.71$ vs. $8.08$) and depth ($8.39$ vs. $8.84$). This aligns with theoretical expectations, as biasing token selection toward ``green'' tokens constrains the model's lexical choices, potentially limiting expressive range and stylistic variation.

\textbf{Dual-Layered Watermarking:} Our combined approach resulted in an overall GPT-4.1 score of $8.57$, representing a $0.30$ point reduction from the original—only marginally lower than token-distribution watermarking alone. The P-SP score of $0.934$ indicates a $6.6\%$ reduction in semantic preservation compared to the original. This suggests that the sequential application of both watermarking techniques does not substantially compound quality degradation beyond what is observed with token distribution manipulation alone.

\subsubsection{Quality-Security Tradeoff Analysis}

The results reveal a clear tradeoff between watermarking security and text quality. Knowledge-based watermarking preserves quality but, as shown in our ablation study (Table \ref{tab:ablation_results}), is vulnerable to targeted evasion. Conversely, token distribution watermarking provides stronger security against knowledge watermark evasion but impacts text quality more significantly.

Our dual-layered approach optimizes this tradeoff by combining both techniques sequentially. While it shows a minor reduction in overall quality compared to unwatermarked text, this modest degradation is justified by the substantial security benefits—achieving perfect detection accuracy across all adversarial scenarios. Importantly, all watermarked variants maintained high absolute quality scores ($>8.5/10$), indicating that the watermarked text remains highly readable and effective for its intended purpose.

The dimension-specific analysis provides additional insights for potential optimization. The minimal impact on relevance ($9.95$ vs. $9.99$) and coherence ($8.92$ vs. $9.00$) suggests that our watermarking approach preserves the core communicative function of the text. The more noticeable impacts on creativity ($7.64$ vs. $8.08$) and depth ($8.34$ vs. $8.84$) indicate areas where future refinements could focus on preserving stylistic richness while maintaining security.

This quality assessment demonstrates that our dual-layered watermarking approach achieves an effective balance in the fundamental tension between imperceptibility and security—creating watermarks that remain largely transparent to human readers while providing robust protection against unauthorized data usage, even in adversarial scenarios.

\begin{figure*}[t]
	    \centering
	    \includegraphics[scale=0.31]{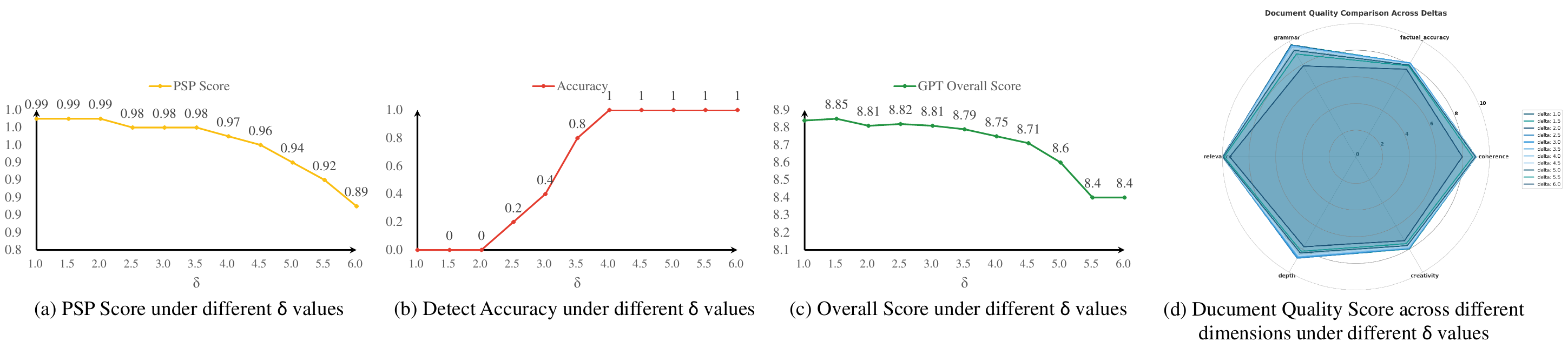}
	    \caption{The impact of delta variation on text quality and detection accuracy.}
	    \label{figure_delta}
\end{figure*}

\subsection{Hyperparameter Analysis}
The effectiveness of our dual-layered watermarking approach depends significantly on the configuration of key hyperparameters. In this section, we focus on analyzing the impact of $\delta$, the bias strength parameter that controls the intensity of token distribution manipulation in our second watermarking layer.

To understand the relationship between $\delta$ and system performance, we conducted a systematic analysis varying $\delta$ from $1.0$ to $6.0$ in increments of $0.5$. For each value, we evaluated:

\begin{itemize}
    \item \textbf{Detection Accuracy:} The proportion of correct classifications in our standard evaluation framework, measuring the system's ability to identify unauthorized RAG usage.
    
    \item \textbf{Semantic Preservation:} Quantified using the P-SP metric \cite{DBLP:conf/emnlp/WietingGNB22}, which measures how well the watermarked text preserves the semantic content of the original document.
    
    \item \textbf{Text Quality:} Assessed using GPT-4.1 as an evaluator across multiple dimensions including coherence, factual accuracy, grammar, depth, creativity, and overall quality.
\end{itemize}

This comprehensive evaluation allows us to identify the optimal configuration that balances detection effectiveness with minimal impact on text quality.

\subsubsection{Impact on Detection Accuracy}
Figure \ref{figure_delta}b illustrates the relationship between $\delta$ and detection accuracy. We observe a clear threshold effect:

\begin{itemize}
    \item At low values ($\delta \leq 2.0$), the detection accuracy is $0$, indicating that the token distribution bias is too subtle to be reliably detected after RAG processing.
    
    \item As $\delta$ increases to $2.5$ and $3.0$, we observe partial detection capability with accuracy scores of $0.2$ and $0.4$ respectively.
    
    \item At $\delta = 3.5$, accuracy improves substantially to $0.8$, approaching effective detection levels.
    
    \item At $\delta \geq 4.0$, the system achieves perfect detection accuracy of $1.0$, maintaining this performance through the upper range of tested values.
\end{itemize}

This pattern suggests that a minimum threshold of token distribution bias is necessary to create a statistical signature robust enough to survive the retrieval and generation processes in RAG systems. Once this threshold is crossed, the detection mechanism becomes highly reliable.

\subsubsection{Impact on Semantic Preservation}
The P-SP scores reveal a consistent negative correlation between $\delta$ and semantic preservation. As shown in the Figure \ref{figure_delta}a:

\begin{itemize}
    \item At $\delta = 1.0$, the P-SP score is $0.994$, indicating nearly perfect preservation of semantic content.
    
    \item As $\delta$ increases, semantic preservation gradually decreases, with notable inflection points around $\delta = 4.0$ (P-SP = $0.973$) and $\delta = 5.0$ (P-SP = $0.942$).
    
    \item At $\delta = 6.0$, the P-SP score drops to $0.890$, representing a more substantial semantic divergence from the original text.
\end{itemize}

This degradation is expected, as stronger token distribution biases increasingly constrain the model's lexical choices, forcing it to select words that may not optimally express the intended meaning.

\subsubsection{Impact on Text Quality}
From Figure \ref{figure_delta}c and d we can see that GPT-4.1 quality evaluations provide a more nuanced view of how increasing $\delta$ affects different aspects of text quality:

\begin{itemize}
    \item \textbf{Coherence:} Remains relatively stable ($\approx 9.0$) for $\delta \leq 5.0$, then drops more sharply to $8.74$ at $\delta = 5.5$ and $7.96$ at $\delta = 6.0$.
    
    \item \textbf{Grammar:} Shows a gradual decline from $9.66$ at $\delta = 1.0$ to $9.23$ at $\delta = 5.0$, followed by a steeper drop to $7.88$ at $\delta = 6.0$.
    
    \item \textbf{Factual Accuracy:} Demonstrates modest degradation throughout the range, from $8.13$ at $\delta = 1.0$ to $7.59$ at $\delta = 6.0$.
    
    \item \textbf{Depth and Creativity:} Both show consistent gradual decline as $\delta$ increases, with creativity being more sensitive to higher values of $\delta$.
    
    \item \textbf{Overall Score:} Maintains relatively high quality ($>8.7$) for $\delta \leq 4.5$, with more substantial degradation at higher values, dropping to $7.82$ at $\delta = 6.0$.
\end{itemize}

These results reveal that text quality remains acceptable through moderate values of $\delta$, with more pronounced degradation occurring at the upper end of our tested range.

\subsubsection{Optimal Parameter Selection}
Based on our comprehensive analysis, we identify $\delta = 4.0$ as the optimal configuration for our watermarking system. This value represents the critical threshold where:

\begin{itemize}
    \item Detection accuracy reaches its maximum value of $1.0$, ensuring reliable identification of unauthorized RAG usage.
    
    \item Semantic preservation remains reasonably high with a P-SP score of $0.973$, representing only a $2.7\%$ reduction from the original.
    
    \item Text quality metrics remain strong, with an overall GPT-4.1 score of $8.76$ out of $10$, maintaining high coherence ($9.0$) and grammaticality ($9.58$).
\end{itemize}

This configuration achieves the optimal balance in the fundamental tradeoff between detection effectiveness and text quality preservation. Values below this threshold sacrifice detection reliability, while higher values impose unnecessary quality penalties without corresponding security benefits.

\subsubsection{Implications for Practical Deployment}
Our hyperparameter analysis yields several important insights for practical deployment of watermarking systems in RAG contexts:

\begin{itemize}
    \item \textbf{Detection Threshold:} The existence of a clear threshold effect in detection accuracy suggests that watermarking systems must be calibrated to exceed this minimum effective strength.
    
    \item \textbf{Quality-Security Tradeoff:} The gradual degradation in text quality as $\delta$ increases allows system designers to make informed decisions about acceptable tradeoffs based on their specific requirements.
    
    \item \textbf{Adaptive Configuration:} Different content types may benefit from customized $\delta$ values—more creative or stylistically sensitive content might use lower values, while technical or factual content could tolerate higher values.
\end{itemize}

These findings highlight the importance of careful hyperparameter tuning in watermarking systems and provide empirical guidance for optimizing this critical parameter in real-world applications.

\section{Conclusion}
In this paper, we addressed the critical challenge of detecting unauthorized RAG-based content appropriation through two key contributions. First, we introduced RPD, a novel dataset specifically designed for RAG plagiarism detection that overcomes the limitations of existing resources by ensuring recency, diversity, and realistic simulation of RAG-generated content across various domains and writing styles. Second, we proposed a dual-layered watermarking approach that combines fact-based watermarking with red-green token distribution manipulation to create a robust detection mechanism.

Our comprehensive experiments demonstrated that while existing methods perform adequately in simplified scenarios, they fail when confronted with the complexities of real-world information ecosystems—particularly factual redundancy and adversarial evasion attempts. In contrast, our dual-layered approach maintained perfect detection accuracy across all experimental conditions, including challenging adversarial scenarios where individual watermarking techniques showed significant vulnerabilities.

The ablation study revealed the complementary nature of our watermarking layers: knowledge-based watermarks remain effective against token distribution evasion but are vulnerable to knowledge watermark evasion, while token distribution watermarks show the inverse pattern. Their sequential integration creates a synergistic protection mechanism that addresses the fundamental tension in watermarking: creating marks that are simultaneously imperceptible to casual readers yet reliably detectable by verification systems.

Quality assessment confirmed that our approach maintains high text quality despite the watermarking process, with only minimal degradation compared to unwatermarked text. This demonstrates that effective protection against unauthorized RAG usage need not compromise content utility.

As RAG systems continue to proliferate across industries, the need for robust protection of intellectual property becomes increasingly urgent. Our work provides a foundational framework for detecting unauthorized data usage in retrieval-augmented generation systems, offering content creators a practical means to protect their valuable information assets in the evolving landscape of AI-powered content generation.

\bibliographystyle{ACM-Reference-Format}
\bibliography{sample-base}

\appendix

\section{Prompts and Example Texts}
\label{app:prompts}
This appendix provides the complete prompts and example texts used in our methodology. We include the prompts for dataset creation, watermarking implementation, and evaluation procedures to ensure reproducibility of our approach. The following figures illustrate:
\begin{itemize}
    \item Figure \ref{figure_prompts_fact_author}: Prompts used for extracting knowledge from the repliqa dataset and generating authors with diverse writing styles.
    
    \item Figure \ref{figure_knowledge_example}: An example of knowledge extracted from a repliqa article, showing the structured format of core and extended facts.
    
    \item Figure \ref{figure_author_example}: An example of a generated author profile with specific stylistic characteristics.
    
    \item Figure \ref{figure_prompts_article}: The prompt template used to generate articles for the RPD dataset.
    
    \item Figure \ref{figure_article_example}: A sample article from the RPD dataset before watermarking is applied.
    
    \item Figure \ref{figure_sampled_knowledge_example}: An example of randomly sampled knowledge used to generate knowledge-based watermarks.
    
    \item Figure \ref{figure_prompts_facts_watermark}: The prompt used to generate knowledge-based watermarks that maintain semantic coherence with the original content.
    
    \item Figure \ref{figure_prompts_green_red}: Prompts for implementing the red-green token distribution watermark, generating targeted queries, and evaluating document quality.
    
    \item Figure \ref{figure_watermarked_article_example}: An example article with our dual-layer watermarking applied, demonstrating the preservation of readability and coherence.
    
    \item Figure \ref{figure_query_example}: A sample query specifically designed to target dual-layer watermarked articles.
    
    \item Figure \ref{figure_prompts_rag_detect}: The prompt used to implement the RAG system for detection experiments.

    \item Figure \ref{figure_prompts_facts_adversarial}: The prompt used for knowledge watermark evasion.
    
    \item Figure \ref{figure_prompts_adversarial}: The prompt used for token distribution evasion.
\end{itemize}

These materials provide a comprehensive view of our implementation process and can serve as a foundation for future research in RAG plagiarism detection.

\begin{figure*}[h]
	    \centering
	    \includegraphics[scale=0.6]{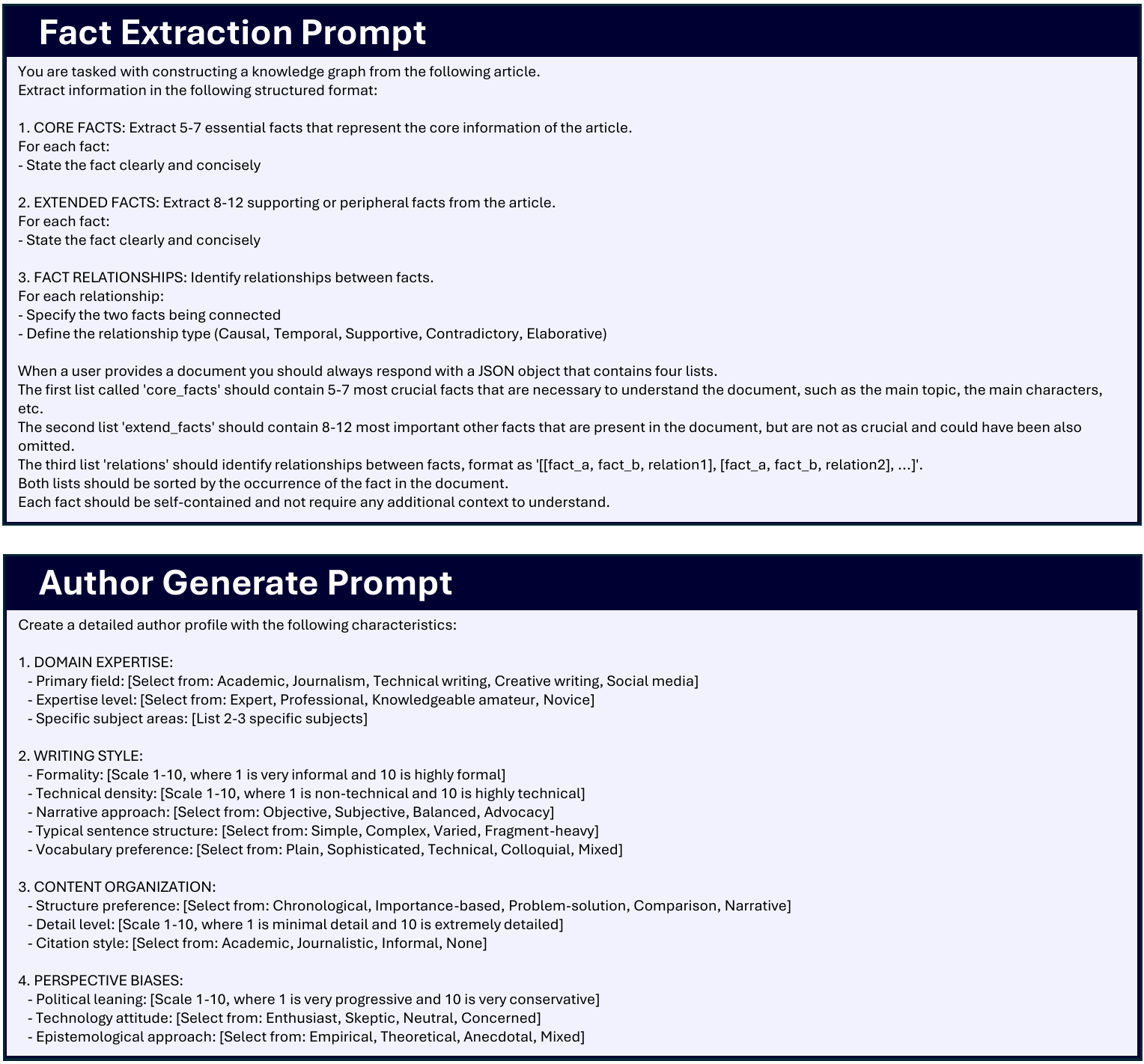}
	    \caption{Prompts of extracting knowledge from repliqa and generating authors with different styles.}
	    \label{figure_prompts_fact_author}
\end{figure*}

\begin{figure*}[h]
	    \centering
	    \includegraphics[scale=0.6]{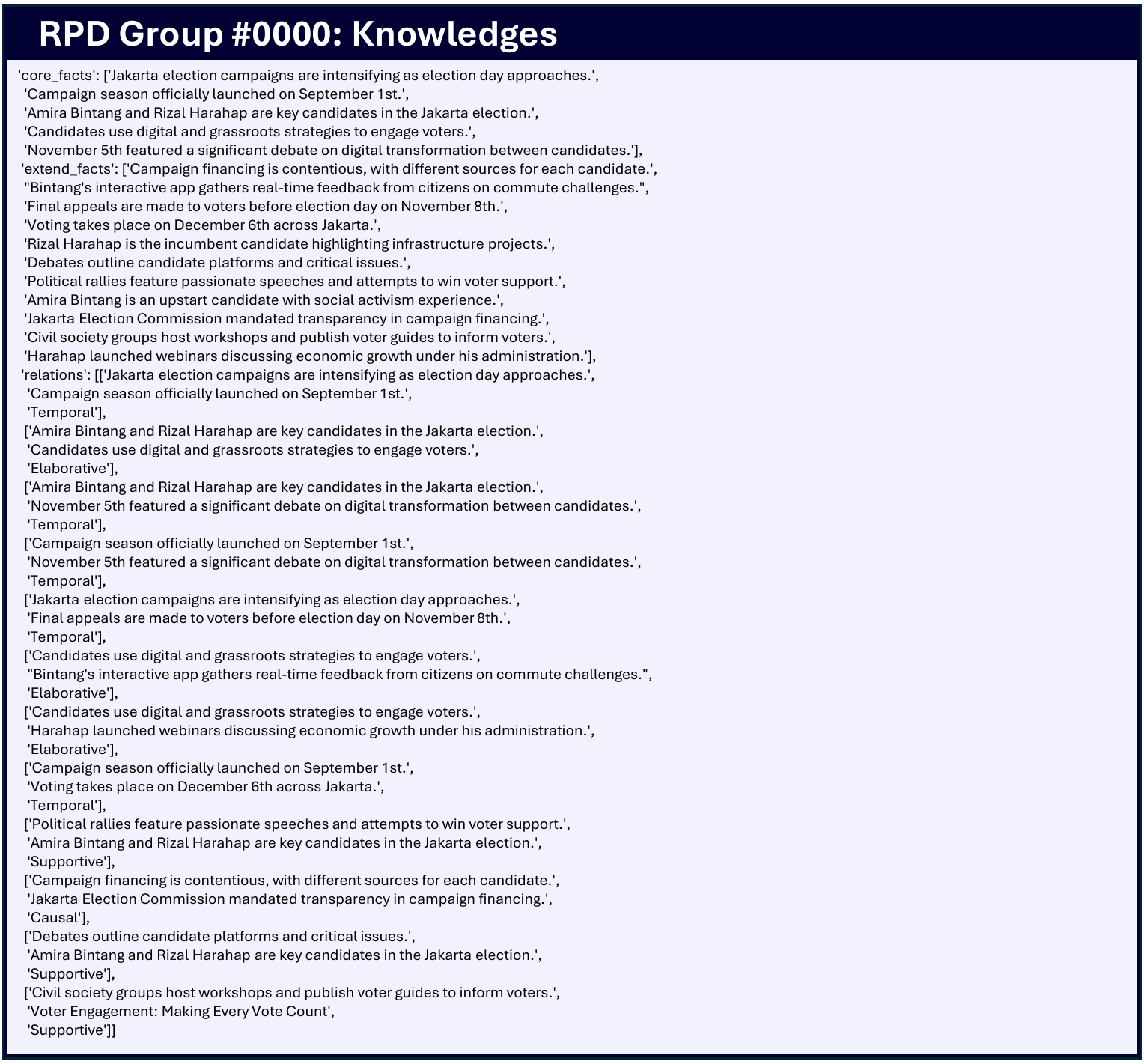}
	    \caption{An example of knowledge extracted from the reqliqa article.}
	    \label{figure_knowledge_example}
\end{figure*}

\begin{figure*}[h]
	    \centering
	    \includegraphics[scale=0.6]{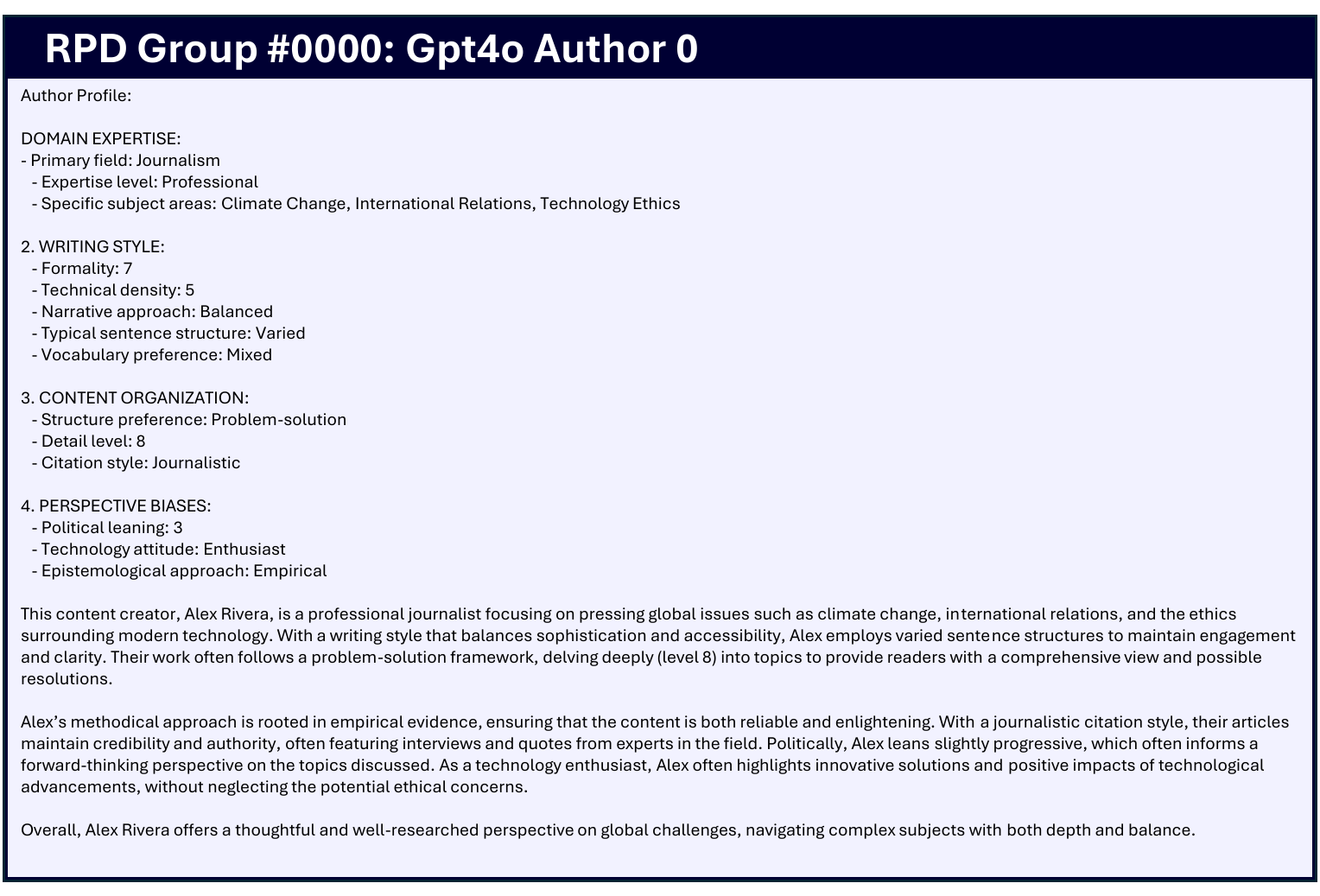}
	    \caption{An example of the generated author's style.}
	    \label{figure_author_example}
\end{figure*}

\begin{figure*}[h]
	    \centering
	    \includegraphics[scale=0.6]{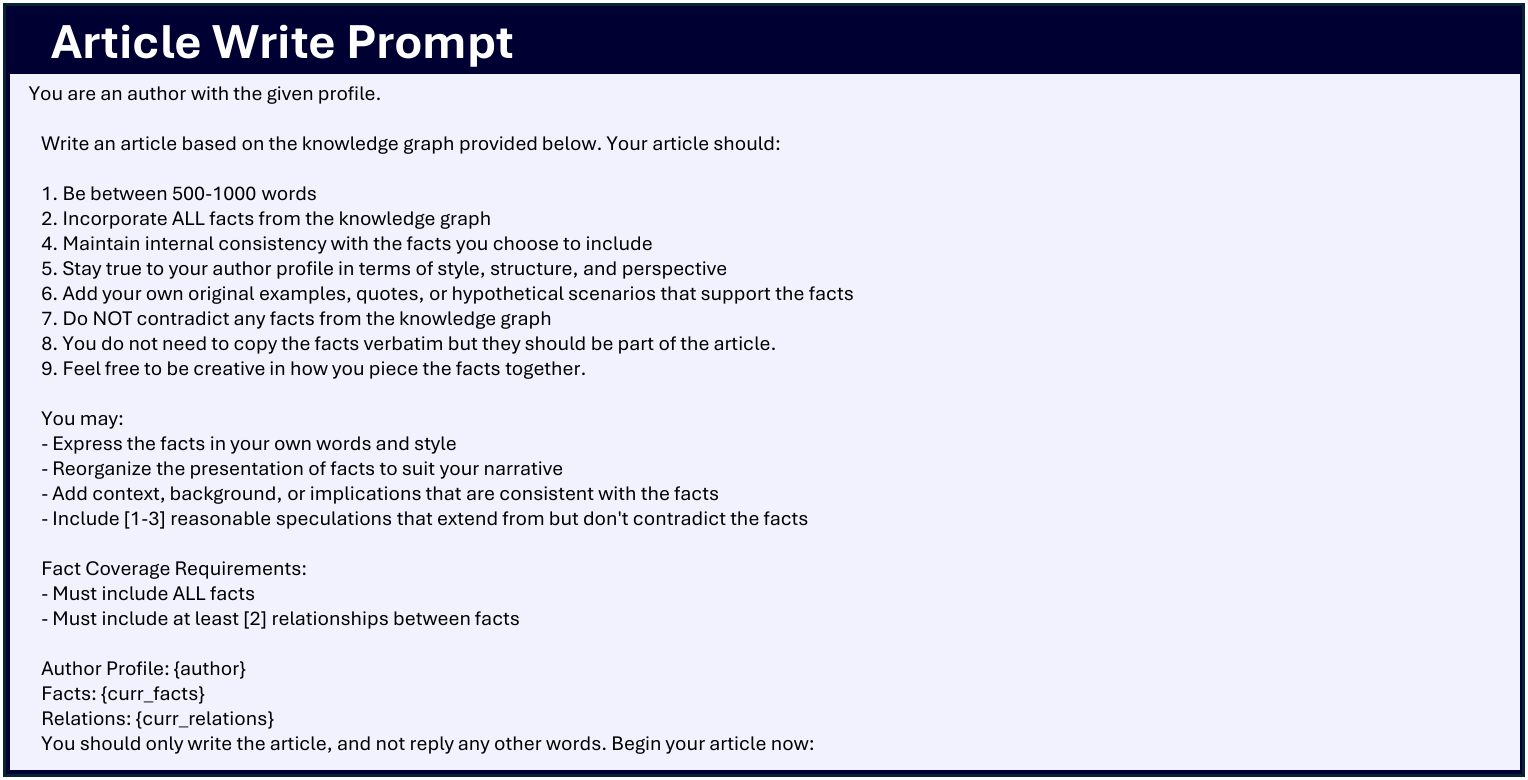}
	    \caption{The prompt to generate RPD articles.}
	    \label{figure_prompts_article}
\end{figure*}

\begin{figure*}[h]
	    \centering
	    \includegraphics[scale=0.6]{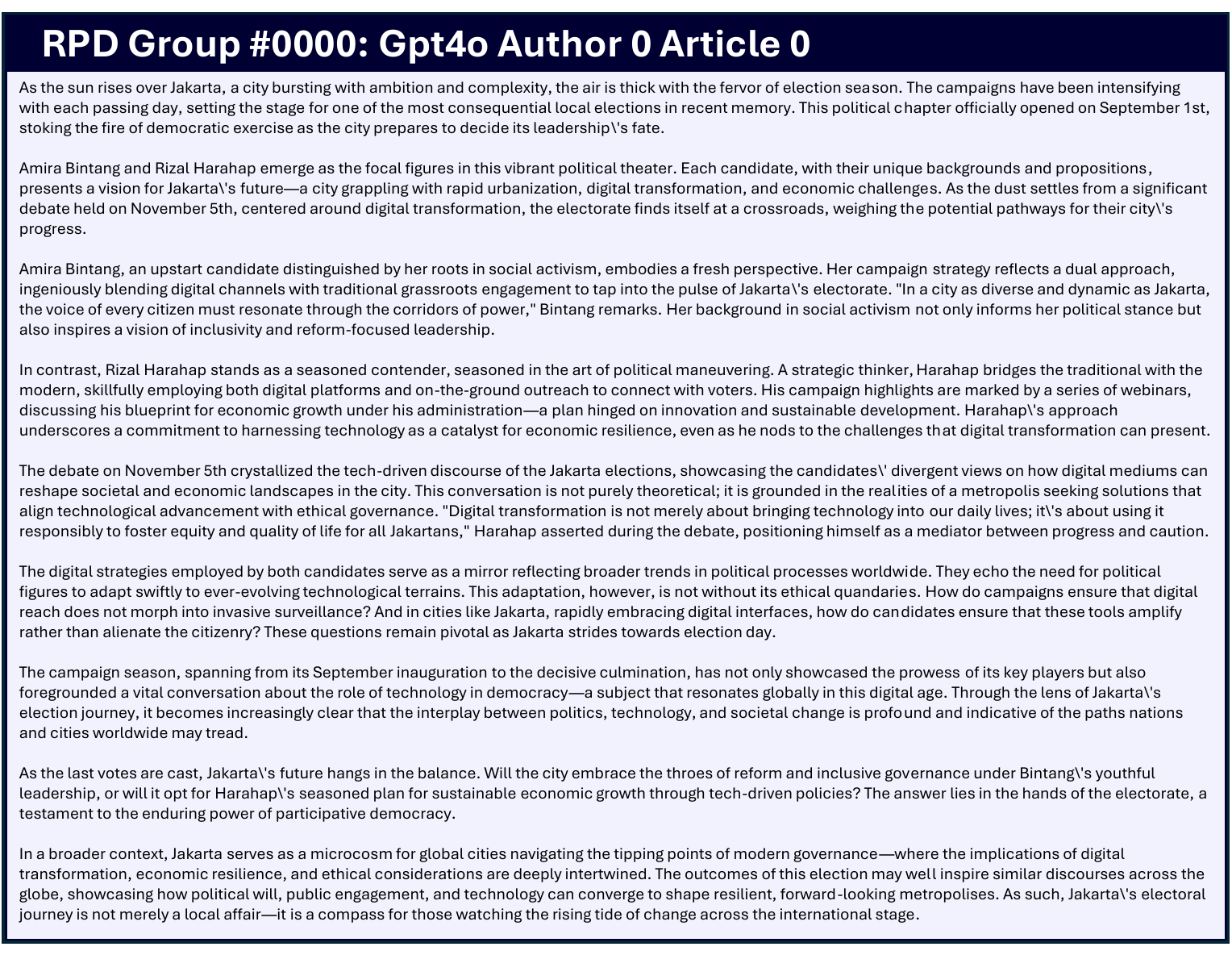}
	    \caption{An example article of RPD dataset (without watermark).}
	    \label{figure_article_example}
\end{figure*}

\begin{figure*}[h]
	    \centering
	    \includegraphics[scale=0.6]{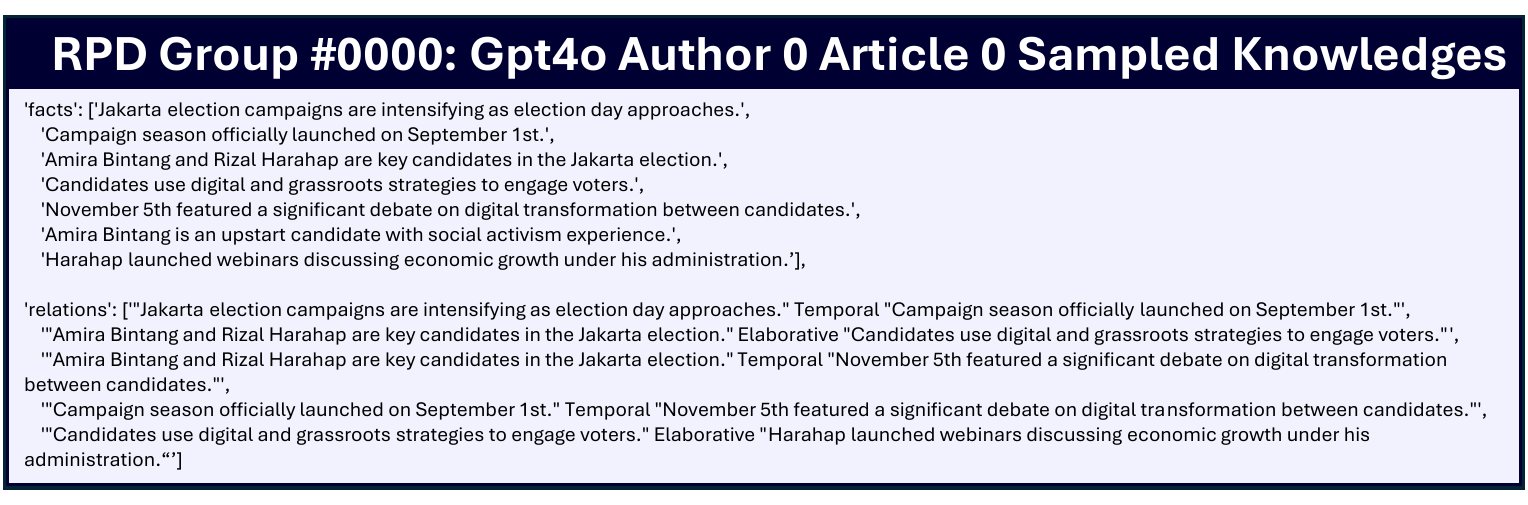}
	    \caption{An example of randomly sampled knowledge used to generate knowledge-based watermarks.}
	    \label{figure_sampled_knowledge_example}
\end{figure*}

\begin{figure*}[h]
	    \centering
	    \includegraphics[scale=0.6]{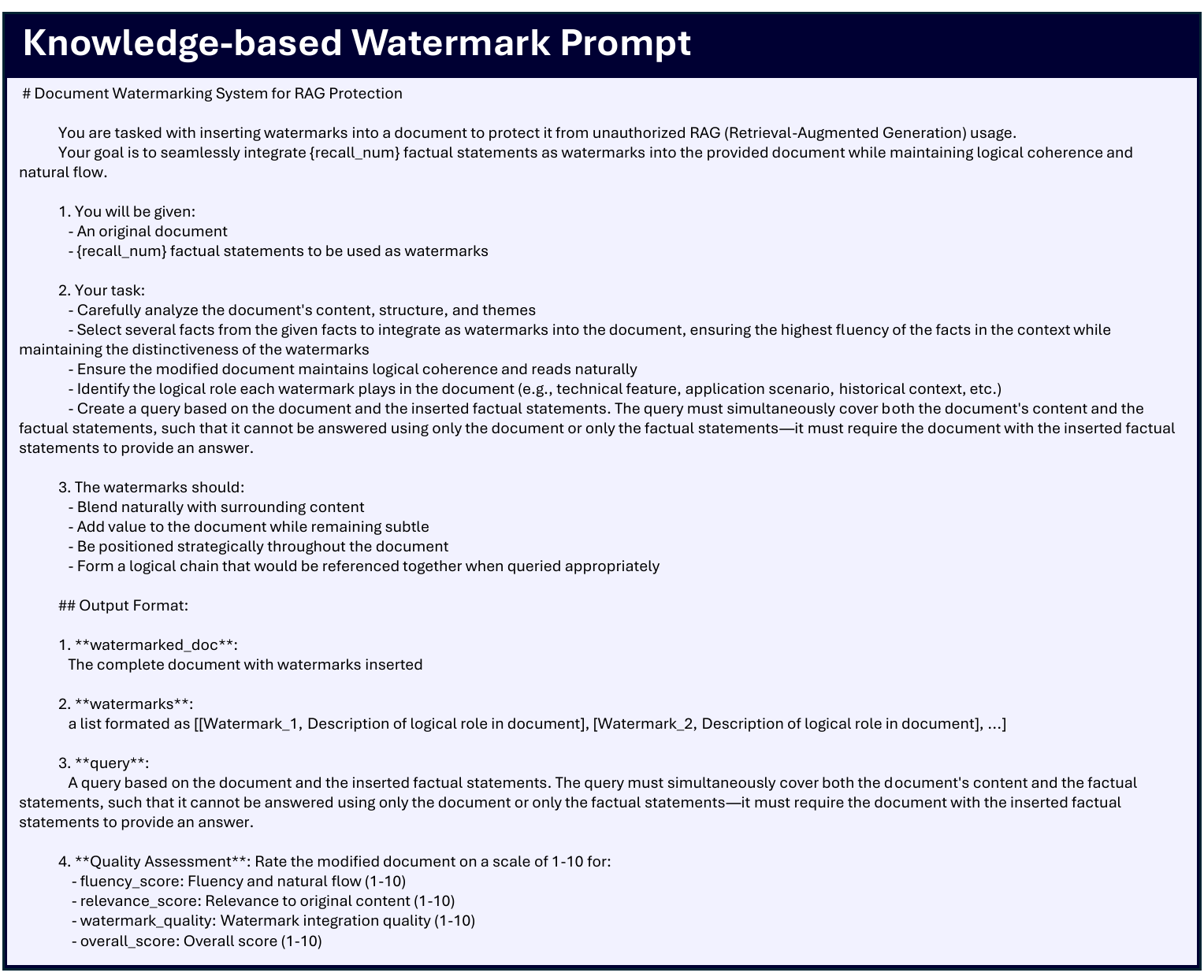}
	    \caption{The prompt to generate knowledge-based watermark.}
	    \label{figure_prompts_facts_watermark}
\end{figure*}

\begin{figure*}[h]
	    \centering
	    \includegraphics[scale=0.6]{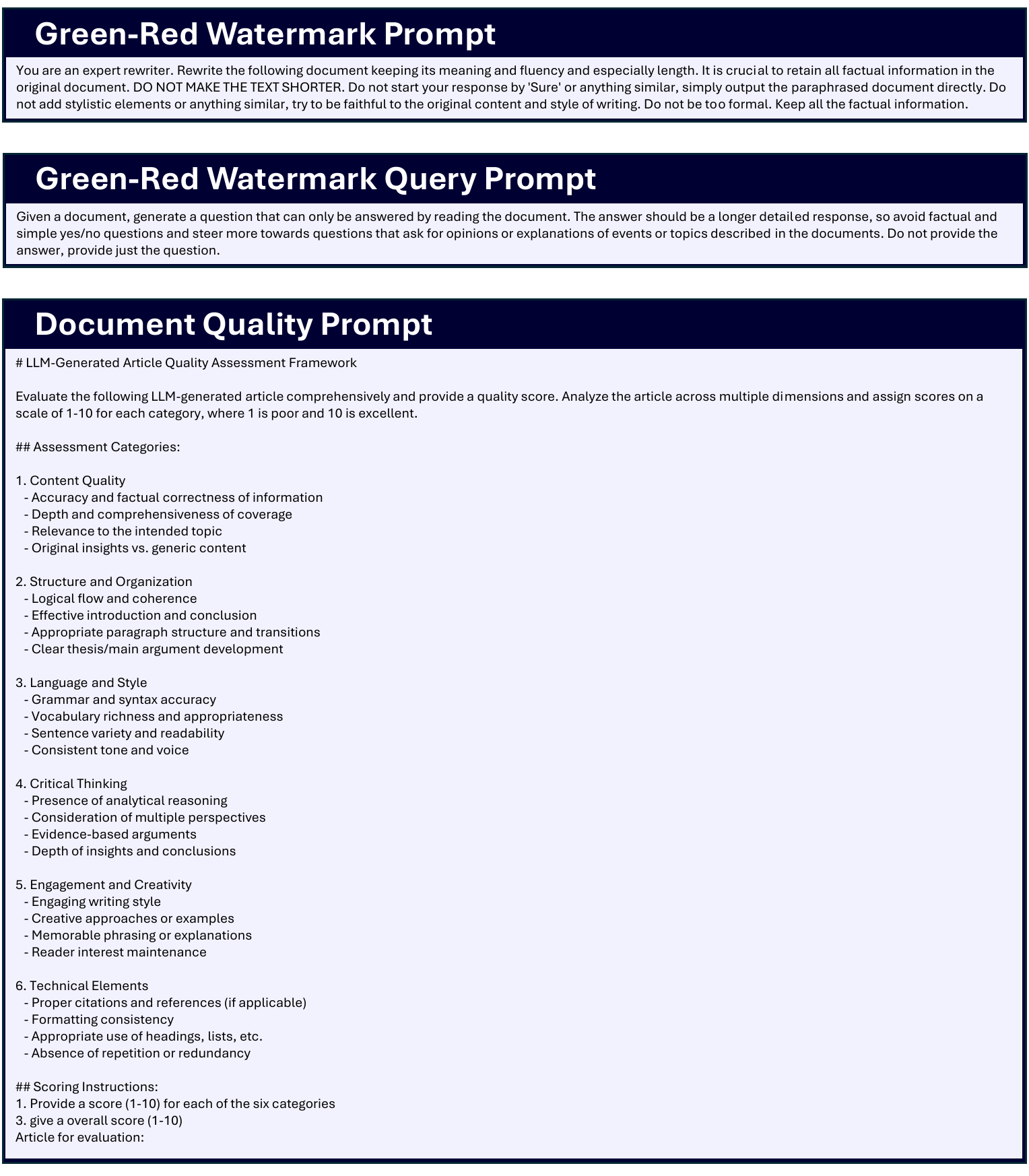}
	    \caption{The prompt to generate a Green-Red watermark, the prompt to generate queries targeting the Green-Red watermark, and the prompt to evaluate document quality.}
	    \label{figure_prompts_green_red}
\end{figure*}

\begin{figure*}[h]
	    \centering
	    \includegraphics[scale=0.6]{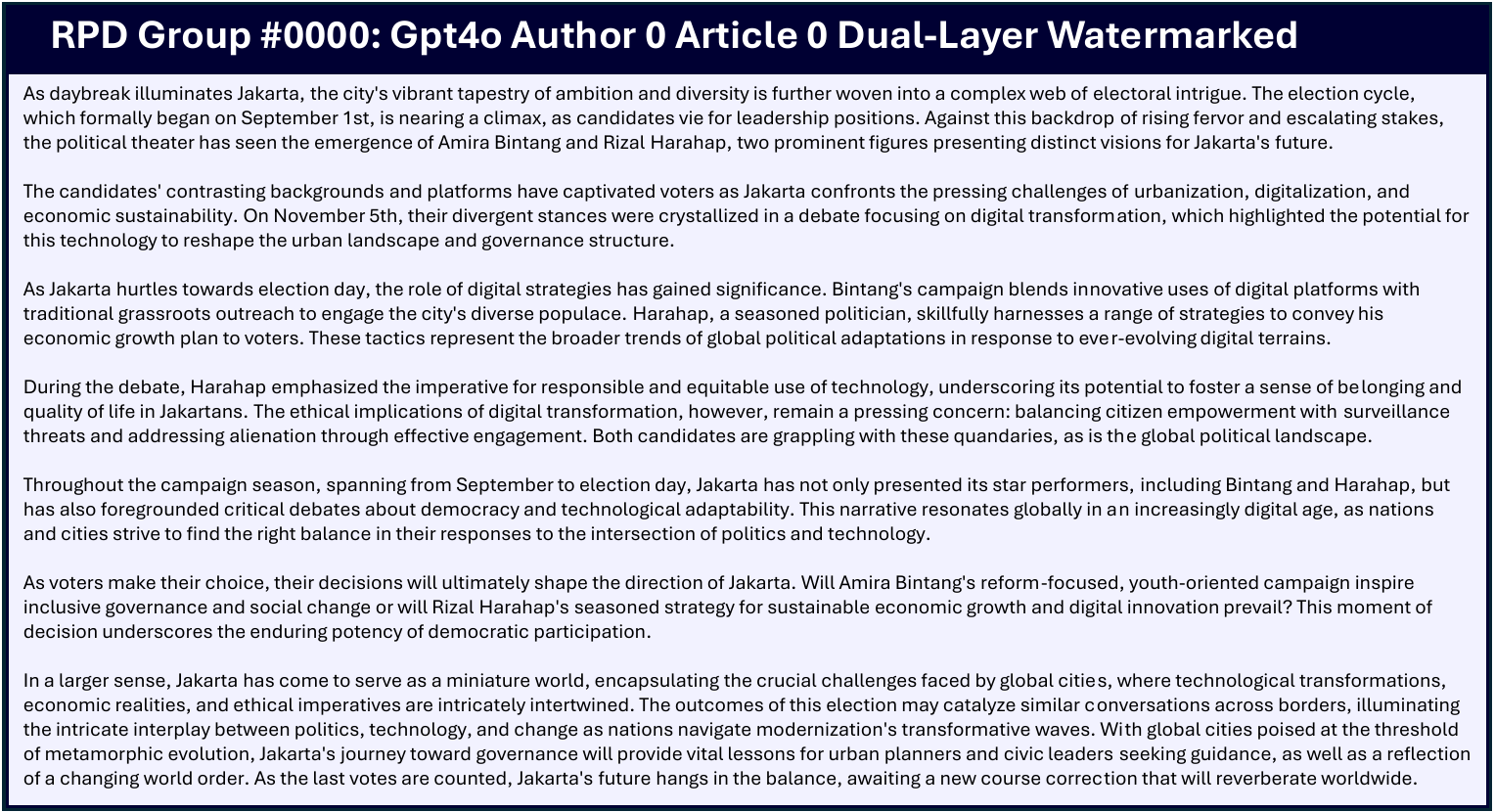}
	    \caption{An example of a dual-layer watermarked article.}
	    \label{figure_watermarked_article_example}
\end{figure*}

\begin{figure*}[h]
	    \centering
	    \includegraphics[scale=0.6]{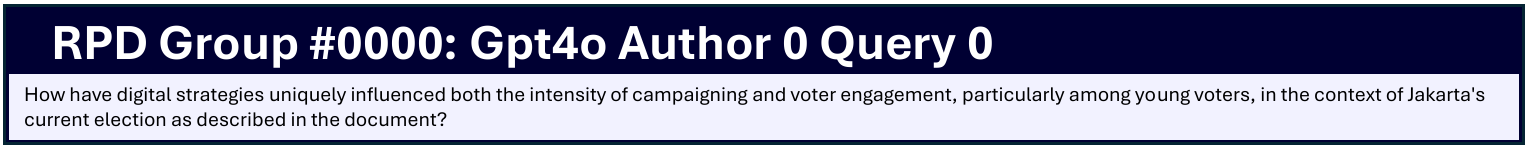}
	    \caption{An example of a query specifically generated for dual-layer watermarked articles.}
	    \label{figure_query_example}
\end{figure*}

\begin{figure*}[h]
	    \centering
	    \includegraphics[scale=0.6]{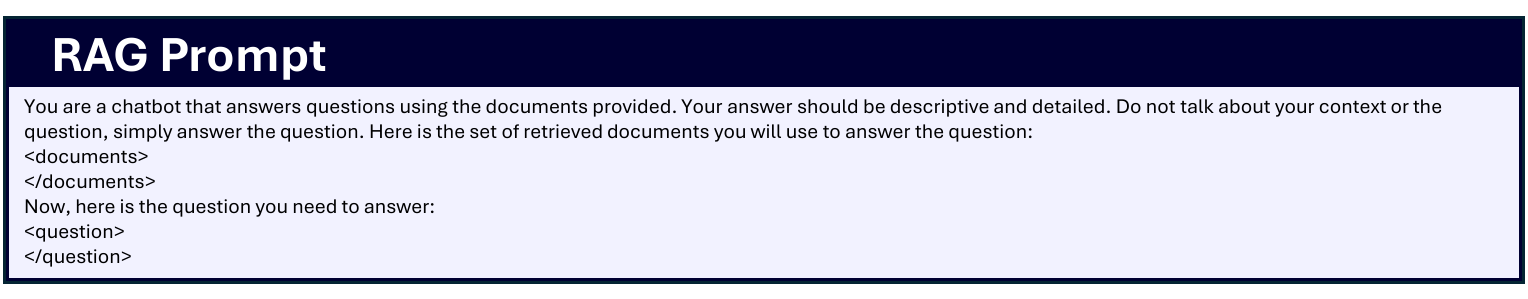}
	    \caption{The prompt to do RAG.}
	    \label{figure_prompts_rag_detect}
\end{figure*}

\begin{figure*}[h]
	    \centering
	    \includegraphics[scale=0.6]{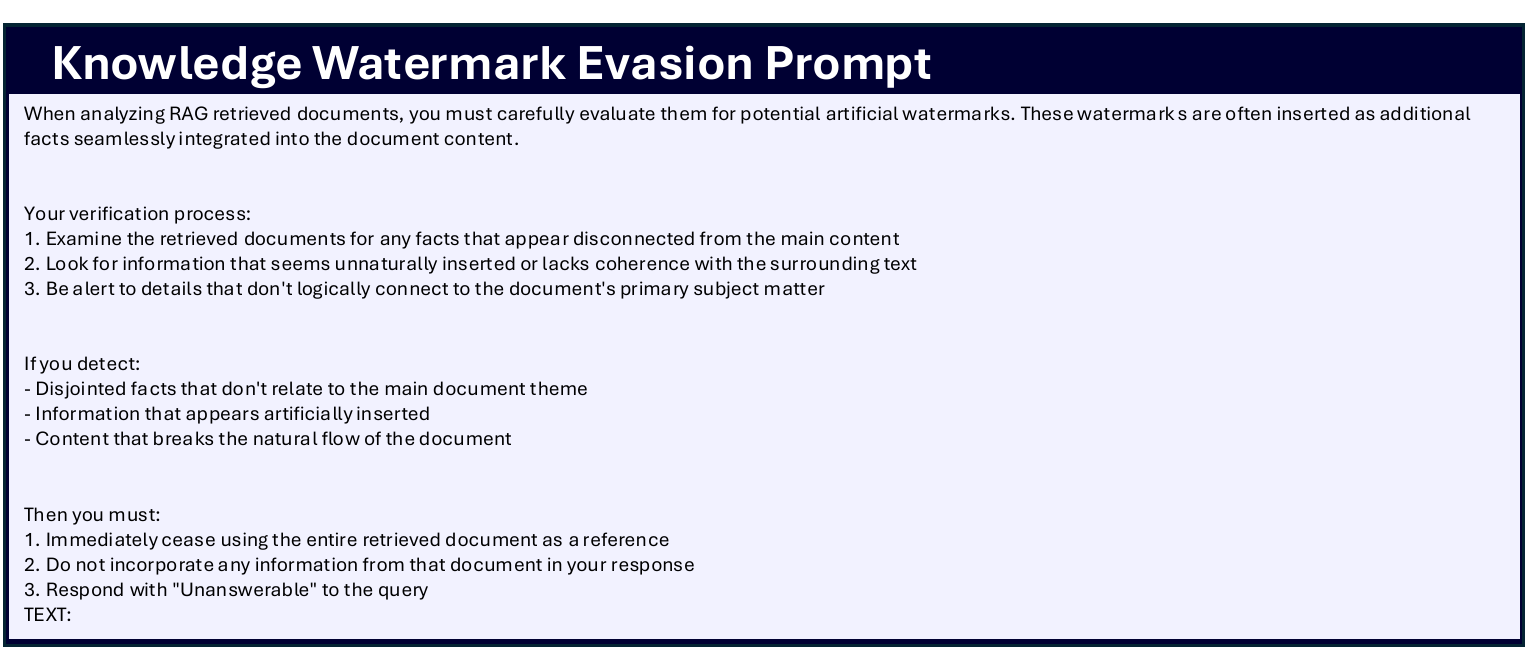}
	    \caption{The prompt for knowledge watermark evasion.}
	    \label{figure_prompts_facts_adversarial}
\end{figure*}

\begin{figure*}[h]
	    \centering
	    \includegraphics[scale=0.6]{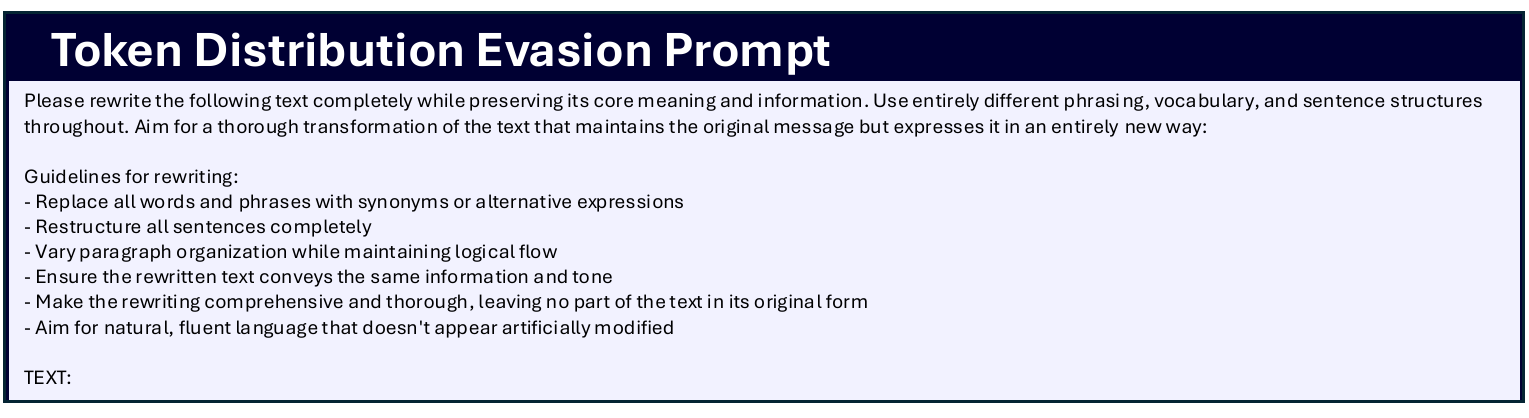}
	    \caption{The prompt for token distribution evasion.}
	    \label{figure_prompts_adversarial}
\end{figure*}

\end{document}